\documentclass[fleqn,11pt,twoside]{article}

\usepackage{amsthm,amssymb, color, xcolor,epsfig, graphics, subfigure} 

\usepackage{amsmath, graphicx, latexsym, lscape,bm}

\makeatletter
\newcommand{\copyrightnote}[2]{{\renewcommand{\thefootnote}{}
 \footnotetext{\small\it
\begin{flushleft}
 \copyright \ #1   #2  
\end{flushleft}}}}

\newcommand{\Name}[1]{\begin{flushleft}
                       \LARGE \bf #1
                       \end{flushleft}\vspace{-3mm}}

\newcommand{\Author}[1]{\begin{flushleft}
                       \it #1 \end{flushleft}}

\newcommand{\Address}[1]{\begin{flushleft}
                       \it #1 \end{flushleft}}

\newcommand{\Date}[1]{\begin{flushleft}
                      \small  \it #1 \end{flushleft}}

\newcommand{\evenhead}{Author \ name}
\newcommand{\oddhead}{Article \ name}

\renewcommand{\@evenhead}{
\hspace*{-3pt}\raisebox{-15pt}[\headheight][0pt]{\vbox{\hbox to \textwidth
{\thepage \hfil \evenhead}\vskip4pt \hrule}}}
\renewcommand{\@oddhead}{
\hspace*{-3pt}\raisebox{-15pt}[\headheight][0pt]{\vbox{\hbox to \textwidth
{\oddhead \hfil \thepage}\vskip4pt\hrule}}}
\renewcommand{\@evenfoot}{}
\renewcommand{\@oddfoot}{}

\long\def\@makecaption#1#2{%
  \vskip\abovecaptionskip
  \sbox\@tempboxa{\small \textbf{#1.}\ \ #2}%
  \ifdim \wd\@tempboxa >\hsize
    {\small \textbf{#1.}\ \ #2}\par
  \else
    \global \@minipagefalse
    \hb@xt@\hsize{\hfil\box\@tempboxa\hfil}%
  \fi
  \vskip\belowcaptionskip}

\newcommand{\JNMPnumberwithin}[3][\arabic]{%
  \@ifundefined{c@#2}{\@nocounterr{#2}}{%
    \@ifundefined{c@#3}{\@nocnterr{#3}}{%
      \@addtoreset{#2}{#3}%
      \@xp\xdef\csname the#2\endcsname{%
        \@xp\@nx\csname the#3\endcsname .\@nx#1{#2}}}}%
}

\renewenvironment{proof}[1][\proofname]{\par
  \normalfont
  \topsep6\p@\@plus6\p@ \trivlist
  \item[\hskip\labelsep\textbf{%
    #1\@addpunct{.}}]\ignorespaces
}{%
  \qed\endtrivlist
}

\newcommand{\resetfootnoterule} {
  \renewcommand\footnoterule{%
  \kern-3\p@
  \hrule\@width.4\columnwidth
  \kern2.6\p@}
}

\renewcommand{\footnoterule}{}

\makeatother

\newtheorem{theorem}{Theorem}

\newtheorem{corollary}[theorem]{Corollary}

\newtheorem{definition}[theorem]{Definition}
\newtheorem{example}[theorem]{Example}

\newtheorem{lemma}[theorem]{Lemma}

\newtheorem{remark}[theorem]{Remark}

\begin{document}

\renewcommand{\evenhead}{ {\LARGE\textcolor{blue!10!black!40!green}{{\sf \ \ \ ]ocnmp[}}}\strut\hfill M B\l aszak and K Marciniak}
\renewcommand{\oddhead}{ {\LARGE\textcolor{blue!10!black!40!green}{{\sf ]ocnmp[}}}\ \ \ \ \  Algebraic curves and separable multi-Hamiltonian systems}

\thispagestyle{empty}
\newcommand{\FistPageHead}[3]{
\begin{flushleft}
\raisebox{8mm}[0pt][0pt]
{\footnotesize \sf
\parbox{150mm}{{Open Communications in Nonlinear Mathematical Physics}\ \ \ \ {\LARGE\textcolor{blue!10!black!40!green}{]ocnmp[}}
\quad Special Issue 2, 2024\ \  pp
#2\hfill {\sc #3}}}\vspace{-13mm}
\end{flushleft}}

\FistPageHead{1}{\pageref{firstpage}--\pageref{lastpage}}{ \ \ }

\strut\hfill

\strut\hfill

\copyrightnote{The author(s). Distributed under a Creative Commons Attribution 4.0 International License}

\begin{center}

{\bf {\large Proceedings of the OCNMP-2024 Conference:\\ 

\smallskip

Bad Ems, 23-29 June 2024}}
\end{center}

\smallskip

\Name{Algebraic curves as a source of separable multi-Hamiltonian systems}

\Author{Maciej B\l aszak$^{\,1}$ and Krzysztof Marciniak$^{\,2}$}

\Address{$^{1}$ Maciej B\l aszak\\Faculty of Physics, Department of of Mathematical Physics and Computer Modelling,\\A. Mickiewicz University, 61-614 Pozna\'{n}, Poland\\\texttt{blaszakm@amu.edu.pl}\\[2mm]
$^{2}$ Department of Science and Technology \\Campus Norrk\"{o}ping, Link\"{o}ping University\\601-74 Norrk\"{o}ping, Sweden\\\texttt{krzma@itn.liu.se}}

\Date{Received January 12, 2024; Accepted February 28, 2024}

\setcounter{equation}{0}

\begin{abstract}

\noindent 
In this paper we systematically consider various ways of generating integrable
and separable Hamiltonian systems in canonical and in non-canonical
representations from algebraic curves on the plane. In particular, we consider
St\"{a}ckel transform between two pairs of St\"{a}ckel systems, generated by
$2n$-parameter algebraic curves on the plane, as well as Miura maps between
St\"{a}ckel systems generated by $(n+N)$-parameter algebraic curves, leading
to multi-Hamiltonian representation of these systems.

\end{abstract}

\label{firstpage}


\section{Introduction\label{s0}}

This paper is devoted to a systematic (in the sense explained below)
construction of various types of Liouville integrable and separable
Hamiltonian systems from algebraic curves.

In \cite{Sklyanin} Sklyanin noted that any Liouville integrable system (that
is a set of $n$ Hamiltonians in involution on a $2n$-dimensional manifold $M$)
separates in a given canonical coordinate system
$(\bm{\lambda},\bm{\mu})\equiv(\lambda_{1},\ldots,\lambda_{n},\mu_{1},\ldots,\mu_{n})$ if and only if there exists $n$ separation relations of the
form
\begin{equation}
\varphi_{i}(\lambda_{i},\mu_{i},h_{1},...,h_{n})=0,\qquad i=1,\ldots,n
\label{0.1}%
\end{equation}
(see also \cite{FP}). Alternatively, one can treat the relations (\ref{0.1})
as an algebraic definition of $n$ commuting, by construction, Hamiltonians
$h_{i}$ on $M$. The canonical variables $(\bm{\lambda},\bm{\mu})$ are then by
construction separation variables for all the Hamilton-Jacobi equations
associated with the Hamiltonians $h_{i}$. This shift of view yields a powerful
way of generating many (in fact, all known in literature) separable
Hamiltonian systems from scratch. This approach has been initiated in
\cite{blaszak2000} and then fruitfully developed in many papers, see for
example \cite{blaszak2005,blaszak2011,SAPM2012} as well as the review of the
subject in \cite{blasz2019}.

In this paper we restrict ourselves to the important subclass of separations
relations (\ref{0.1}) where all $\varphi_{i}$ are the same, $\varphi
_{i}=\varphi$ for all $i$. In such a case the relations (\ref{0.1}) can be
interpreted as $n$ copies of the algebraic curve on the $\lambda$-$\mu$ plane
\begin{equation}
\varphi(\lambda,\mu,h_{1},...,h_{n})=0, \label{0.2}%
\end{equation}
when $(\lambda,\mu)$ are consecutively substituted by the pair of variables
$(\lambda_{i},\mu_{i})$. One reason for restricting our attention to
separation curves (\ref{0.2}) rather than the general separation relations
(\ref{0.1}) is that in the general setting there arise problems with finding
the multi-Hamiltonian formulation of the generated separable systems. Another
reason for this restriction is that it allows us to skip the assumption that
the corresponding coordinates $(\bm{\lambda},\bm{\mu})$ on $M$ are canonical.
Using this approach, in this article we systematize and develop the idea of
constructing various types of finite-dimensional integrable and separable
Hamiltonian systems from \emph{parameter-dependent }planar algebraic curves.
To our knowledge this is the first systematic (albeit certainly not complete)
investigation of separation curves depending on more than $n$ parameters.
Relations between integrable systems and \thinspace$n$-parameter hyperelliptic
curves were extensively investigated for example in \cite{vh2001} (see also
references there).

Below we present the structure of the article and highlight all the new results.

In Section \ref{s1} we establish a number of facts for Poisson structures in
$2$ dimensions and of monomial type. We first establish Lemma \ref{D} where we
find all Darboux coordinates associated with the Poisson tensor (\ref{1.2}%
)$\,\ $on the plane with $c$ of the monomial form $c=\lambda^{\alpha}%
\mu^{\beta}$ and then we prove Lemma \ref{DD} where we establish canonical
maps between arbitrary pair of Darboux coordinates for $\pi$. These results
will be necessary for establishing results of Section 3.

In the first subsection of Section \ref{s2} we prove (Theorem \ref{KOM}) that
the Hamiltonians $h_{i}$ obtained by algebraically solving $n$ copies of
(\ref{0.2}) constitute a Liouville integrable system not only if the
corresponding coordinates $(\bm{\lambda},\bm{\mu})$ on $M$ are canonical, but
in a more general case when the Poisson operator $\pi$ has in the variables
$(\bm{\lambda},\bm{\mu})$ the form (\ref{2.3}). This is a simple
generalization of the previously known result (see for example
\cite{blasz2019}). The second subsection of Section \ref{s2} contains basic
information on separable Hamiltonian systems, in particular of St\"{a}ckel
type. This subsection is necessary for the self-consistency of the article.

In Section \ref{s3} we use the results of Section \ref{s1} to show that each
Liouville integrable Hamiltonian system generated by an algebraic curve
(\ref{0.2}) and by the non-canonical Poisson tensor (\ref{1.2}) can also be
generated by a one-parameter family of algebraic curves and the corresponding
Poisson tensors in canonical form. Each class represents thus the same
dynamical system written in different Darboux coordinates, related with each
other by appropriate canonical transformations. We further specify these
results for the case of monomial Poisson structures (with $c(\lambda_{i}%
,\mu_{i})=\lambda_{i}^{\alpha}\mu_{i}^{\beta}$), see formulas (\ref{3.2})
yielding explicit transformation to Darboux coordinates in this case.

In Section \ref{s4} we consider separable systems generated by algebraic
curves depending on a set of $n+n$ rather than $n$ parameters. Each such curve
leads then to two distinct integrable Hamiltonian systems. Using the known
theory (\cite{serg2008,SAPM2012}) we prove that these systems are related by a
St\"{a}ckel transform and we also show how solutions of these two systems are
related by a reciprocal (multi-time) transformations. We also specify these
results to the case of St\"{a}ckel systems. These results are new.

In the final Section \ref{s5} we investigate yet another possibility of
generating integrable and separable Hamiltonian systems from algebraic curves.
We consider algebraic curves (\ref{5.1}) with the block-type structure given
by (\ref{5.1a}) and (\ref{5.1b}), leading to families of integrable and
separable Hamiltonian systems that can be related (due to Theorem
\ref{obvious}) with each other by a finite-dimensional analogue of Miura maps.
These finite-dimensional Miura maps yield in turn multi-Hamiltonian
formulation of the obtained integrable systems, as presented in Theorem
\ref{obvious2}. These are new results that generalize the particular results
obtained earlier in \cite{blaszak2009} and in \cite{marciniak2023}. The
section is concluded by some examples. Theorem \ref{obvious} is proven in
Appendix, due to a rather technical nature of the proof.

\section{Poisson tensors on $2$-dimensional phase space\label{s1}}

In this section we consider Poisson structures in $2$ dimensions (on a
$\lambda$-$\mu$ plane) of a monomial type and find all their Darboux
coordinates that can be obtained from the coordinates $(\lambda,\mu)$ by
monomial transformations. We also find a general map between arbitrary Darboux
coordinate systems of monomial type.

Let us thus consider a ($\lambda,\mu)$ plane $P$. A Poisson tensor $\pi$ on
$P$ is a bi-vector with vanishing Schouten-Nijenhuis bracket. The Poisson
tensor $\pi$ must be of co-rank zero since $\dim P=2$ . It defines a Poisson
bracket on the plane:
\begin{equation}
\{f,g\}_{\pi}:=\pi(df,dg),\ \ \ \ \ \ f,g\in C^{\infty}(P). \label{1.1}%
\end{equation}

\begin{lemma}
The most general Poisson tensor $\pi$ on $P$ is of the form%
\begin{equation}
\pi=c(\lambda,\mu)\frac{\partial}{\partial\lambda}\wedge\frac{\partial
}{\partial\mu}\text{,}\ \text{\ with matrix representation }\pi=\left(
\begin{array}
[c]{cc}%
0 & c(\lambda,\mu)\\
-c(\lambda,\mu) & 0
\end{array}
\right)  ,\ \ \ c\in C^{2}(P). \label{1.2}%
\end{equation}

\end{lemma}

\begin{proof}
The necessary and sufficient condition for being Poisson tensor is a Jacobi
identity%
\begin{equation}
\{f,\{g,h\}_{\pi}\}_{\pi}+\{g,\{h,f\}_{\pi}\}_{\pi}+\{h,\{f,g\}_{\pi}\}_{\pi
}=0. \label{1.3}%
\end{equation}
By a direct computation one can verify the identity (\ref{1.3}) for tensor
(\ref{1.2}), where
\begin{equation}
\{f,g\}_{\pi}=\left(  \frac{\partial f}{\partial\lambda}\frac{\partial
g}{\partial\mu}-\frac{\partial f}{\partial\mu}\frac{\partial g}{\partial
\lambda}\right)  c. \label{1.4}%
\end{equation}

\end{proof}

Now, let us change the parametrization of the plane $(\lambda,\mu
)\longrightarrow(\bar{\lambda},\bar{\mu})$ given by%
\begin{equation}
\bar{\lambda}=a(\lambda,\mu),\ \ \ \ \ \bar{\mu}=b(\lambda,\mu) \label{1.5}%
\end{equation}
and such that new coordinates are Darboux (canonical) coordinates for $\pi$,
i.e. $c(\bar{\lambda},\bar{\mu})=1$. It means that following condition
\begin{equation}
\left(  \frac{\partial a}{\partial\lambda}\frac{\partial b}{\partial\mu}%
-\frac{\partial a}{\partial\mu}\frac{\partial b}{\partial\lambda}\right)  c=1
\label{1.6}%
\end{equation}
has to be fulfilled for a pair of unknown functions $a$ and $b$.

Consider a particular but relevant subclass of Poisson tensors (\ref{1.4})
defined by functions $c$ in the monomial form $c=\lambda^{\alpha}\mu^{\beta}$
for fixed $\alpha$,$\beta\in\mathbb{R}$. Let us now search for transformations
to Darboux coordinates of (\ref{1.2}) within the following ansatz
\begin{equation}
\ \bar{\lambda}=\lambda^{\bar{\alpha}_{1}}\mu^{\bar{\alpha}_{2}}%
,\ \ \ \ \bar{\mu}=\lambda^{\bar{\beta}_{1}}\mu^{\bar{\beta}_{2}}, \label{1.7}%
\end{equation}
where $\bar{\alpha}_{1},\bar{\alpha}_{2},\bar{\beta}_{1},\bar{\beta}_{2}%
\in\mathbb{R}$. This ansatz has the following inverse%
\begin{equation}
\lambda=\bar{\lambda}^{\frac{\bar{\beta}_{2}}{\bar{\alpha}_{1}\bar{\beta}%
_{2}-\bar{\alpha}_{2}\bar{\beta}_{1}}}\bar{\mu}^{-\frac{\bar{\alpha}_{2}}%
{\bar{\alpha}_{1}\bar{\beta}_{2}-\bar{\alpha}_{2}\bar{\beta}_{1}}}%
,\ \ \ \ \mu=\bar{\lambda}^{-\frac{\bar{\beta}_{1}}{\bar{\alpha}_{1}\bar
{\beta}_{2}-\bar{\alpha}_{2}\bar{\beta}_{1}}}\bar{\mu}^{\frac{\bar{\alpha}%
_{1}}{\bar{\alpha}_{1}\bar{\beta}_{2}-\bar{\alpha}_{2}\bar{\beta}_{1}}}.
\label{1.7a}%
\end{equation}
A direct calculation using the condition (\ref{1.6}) shows that the
transformation (\ref{1.7}) turns $\pi\,$into canonical form if and only if%
\begin{equation}
\bar{\alpha}_{1}+\ \bar{\beta}_{1}+\alpha=1,\text{ \ }\bar{\alpha}_{2}%
+\ \bar{\beta}_{2}+\beta=1,\text{ \ }\bar{\alpha}_{1}\bar{\beta}_{2}%
-\bar{\alpha}_{2}\bar{\beta}_{1}=1, \label{rnka}%
\end{equation}
which leads to the following lemma.

\begin{lemma}
\label{D}The coordinates $\left(  \ \bar{\lambda},\ \bar{\mu}\right)  $ given
by (\ref{1.7}) are Darboux (canonical) coordinates for $\pi$ if and only if
the parameters $\bar{\alpha}_{1},\bar{\alpha}_{2},\bar{\beta}_{1},\bar{\beta
}_{2}$ are given by, for $\alpha\neq1$
\begin{equation}
\bar{\alpha}_{1}=\bar{\alpha}_{1},\ \ \ \bar{\alpha}_{2}=\frac{\bar{\alpha
}_{1}(1-\beta)-1}{1-\alpha},~\ \ \ \bar{\beta}_{1}=1-\alpha-\bar{\alpha}%
_{1},\ \ \ \ \bar{\beta}_{2}=\frac{(1-\beta)(1-\alpha-\bar{\alpha}_{1}%
)+1}{1-\alpha},\ \ \ \ \label{1.8}%
\end{equation}
and for $\beta\neq1$%
\begin{equation}
\bar{\alpha}_{1}=\frac{(1-\alpha)(1-\beta-\ \bar{\beta}_{2})+1}{1-\beta
},\ \ \ \bar{\alpha}_{2}=1-\beta-\ \bar{\beta}_{2},\ \ \ \ \bar{\beta}%
_{1}=\frac{\ \bar{\beta}_{2}(1-\alpha)-1}{1-\beta},\ \ \ \ \bar{\beta}%
_{2}=\bar{\beta}_{2}. \label{1.9}%
\end{equation}

\end{lemma}

For $\alpha=\beta=1$ there is no transformations to Darboux coordinates of
(\ref{1.2}) within the ansatz (\ref{1.7}). An example of transformation that
leads to Darboux coordinates in this case is $\bar{\lambda}=\ln\left\vert
\lambda\right\vert ,\bar{\mu}=\ln\left\vert \mu\right\vert .$

Notice that the solutions (\ref{1.8})\ are one-parameter, parametrized by
$\bar{\alpha}_{1}$; similarly, the solutions (\ref{1.9}) are one-parameter,
and parametrized by $\bar{\beta}_{2}$. Note also that the last equation in
(\ref{rnka}) means that not only the map (\ref{1.7}) but also its inverse
(\ref{1.7a}) are in this case polynomial maps.

In the special case that $\alpha=\beta=0$, (that is if the original variables
$(\lambda,\mu)$ are already canonical for $\pi$) the transformation
(\ref{1.7}) given by (\ref{1.8}) represents a one-parameter family of
canonical transformations
\begin{equation}
\bar{\lambda}=\lambda^{\bar{\alpha}_{1}}\mu^{\bar{\alpha}_{1}-1}%
,\ \ \ \ \bar{\mu}=\lambda^{1-\bar{\alpha}_{1}}\mu^{2-\bar{\alpha}_{1}},
\label{1.11}%
\end{equation}
(parametrized by $\bar{\alpha}_{1}$) with the inverse%
\begin{equation}
\ \lambda=\bar{\lambda}^{2-\bar{\alpha}_{1}}\bar{\mu}^{1-\bar{\alpha}%
_{1}},\ \ \ \mu=\bar{\lambda}^{\bar{\alpha}_{1}-1}\bar{\mu}^{\bar{\alpha}_{1}%
}. \label{1.11a}%
\end{equation}
Applying Lemma \ref{D} to two different sets of Darboux coordinates: $\left(
\ \bar{\lambda},\ \bar{\mu}\right)  $ and $(\tilde{\lambda},\tilde{\mu})$,
given by (\ref{1.8}) with $\bar{\alpha}_{1}$and $\tilde{\alpha}_{1}$
respectively, we arrive at the following lemma.

\begin{lemma}
\label{DD}Assume, that $(\bar{\lambda},\bar{\mu})$ and $(\tilde{\lambda
},\tilde{\mu})$ are two different sets of Darboux coordinates for the same
Poisson tensor (\ref{1.2}) with $c=\lambda^{\alpha}\mu^{\beta}$ related to
$(\lambda,\mu)$ by two different solutions (\ref{1.8}) given by $\bar{\alpha
}_{1}$ and by $\tilde{\alpha}_{1}$ respectively. Then the map $(\bar{\lambda
},\bar{\mu})$\ $\longrightarrow(\tilde{\lambda},\tilde{\mu})$ is canonical and
takes the form
\begin{equation}
\bar{\lambda}=\tilde{\lambda}^{\alpha_{1}}\tilde{\mu}^{\alpha_{1}%
-1},\ \ \ \ \bar{\mu}=\tilde{\lambda}^{1-\alpha_{1}}\tilde{\mu}^{2-\alpha_{1}%
},\ \ \ \ \ \ \ \alpha_{1}=1+\frac{\bar{\alpha}_{1}-\tilde{\alpha}_{1}%
}{1-\alpha}\text{, \ \ \ }\alpha\neq1, \label{1.12}%
\end{equation}%
\begin{equation}
\bar{\lambda}=\tilde{\lambda}^{\alpha_{2}}\tilde{\mu}^{\alpha_{2}%
-1},\ \ \ \ \bar{\mu}=\tilde{\lambda}^{1-\alpha_{2}}\tilde{\mu}^{2-\alpha_{2}%
},\ \ \ \ \ \ \ \alpha_{2}=1+\frac{\bar{\alpha}_{2}-\tilde{\alpha}_{2}%
}{1-\beta}\text{, \ \ \ }\beta\neq1. \label{1.121}%
\end{equation}

\end{lemma}

The proof is again obtained by direct calculations. The inverse of
(\ref{1.12}) is of the form%
\begin{equation}
\ \ \tilde{\lambda}=\bar{\lambda}^{2-\alpha_{1}}\bar{\mu}^{1-\alpha_{1}%
},\ \ \ \ \bar{\mu}=\bar{\lambda}^{\alpha_{1}-1}\bar{\mu}^{\alpha_{1}},
\label{1.12a}%
\end{equation}
while the inverse of (\ref{1.121}) is%
\[
\ \ \tilde{\lambda}=\bar{\lambda}^{2-\alpha_{2}}\bar{\mu}^{1-\alpha_{2}%
},\ \ \ \ \bar{\mu}=\bar{\lambda}^{\alpha_{2}-1}\bar{\mu}^{\alpha_{2}}.
\]

\section{From algebraic curves to Liouville integrable and separable
Hamiltonian systems\label{s2}}

\subsection{Liouville integrability}

Here we show how to construct $n$-dimensional Liouville integrable Hamiltonian
systems starting from a single $n$-parameter algebraic curve on a
$(\lambda,\mu)$-plane of the form%
\begin{equation}
\varphi(\lambda,\mu,a_{1},...,a_{n})=0. \label{2.1}%
\end{equation}
Taking $n$ copies of \eqref{2.1} with $(\lambda,\mu)$ consecutively
substituted by the pair of variables $(\lambda_{i},\mu_{i})$ we obtain the
system of $n$ equations
\begin{equation}
\varphi(\lambda_{i},\mu_{i},a_{1},...,a_{n})=0,\qquad i=1,\ldots,n \label{2.2}%
\end{equation}
that is assumed to be solvable with respect to the parameters $a_{k}$ (at
least in some open domain). In result, we obtain $n$ functions (Hamiltonians)
$a_{k}=h_{k}(\bm{\lambda},\bm{\mu})$ on $2n$-dimensional manifold $M$,
parametrized by coordinates $\bm{\lambda }=(\lambda_{1},...,\lambda_{n})$,
$\bm{\mu}=(\mu_{1},...,\mu_{n})$.

In order to turn manifold $M$ into Poisson manifold we take $n$ copies of the
Poisson operator (\ref{1.2}) on the plane and construct the Poisson tensor
$\pi$ on $M$ as follows:
\begin{equation}
\pi=\sum_{i=1}^{n}c(\lambda_{i},\mu_{i})\frac{\partial}{\partial\lambda_{i}%
}\wedge\frac{\partial}{\partial\mu_{i}}, \label{2.3}%
\end{equation}
so that its matrix representation is
\begin{equation}
\pi=\left(
\begin{array}
[c]{cc}%
\bm{0}_{n\times n} & c(\bm{\lambda},\bm{\mu})\\
-c(\bm{\lambda},\bm{\mu}) & \bm{0}_{n\times n}%
\end{array}
\right)  ,\ \ \ \ c(\bm{\lambda},\bm{\mu})=\mathrm{diag}(c(\lambda_{1},\mu
_{1}),...,\ c(\lambda_{n},\mu_{n})). \label{2.4}%
\end{equation}
Below we prove $h_{i}$ generated by (2.2) commute with respect to $\pi$ given
by (\ref{2.3}), which is a natural generalization of the result with $c=1$
that can be found for example in \cite{blasz2019}.

\begin{theorem}
\label{KOM}Hamiltonian functions $h_{i}$ are in involution with respect to
Poisson tensors (\ref{2.3})%
\begin{equation}
\{h_{i},h_{j}\}_{\pi}=0,\ \ \ \ i,j=1,...,n. \label{2.5}%
\end{equation}

\end{theorem}

\begin{proof}
The Hamiltonian functions $h_{i}(\bm{\lambda},\bm{\mu})$ Poisson commute as a
consequence of relations (\ref{2.2}). Indeed, differentiating equations
(\ref{2.2}) with respect to $(\bm{\lambda},\bm{\mu})$ coordinates we get
\[
\frac{\partial\varphi_{k}}{\partial\lambda_{i}}+\sum_{r=1}^{n}\frac
{\partial\varphi_{k}}{\partial a_{r}}\frac{\partial h_{r}}{\partial\lambda
_{i}}=0,\ \ \ \ \ \ \frac{\partial\varphi_{k}}{\partial\mu_{i}}+\sum_{r=1}%
^{n}\frac{\partial\varphi_{k}}{\partial a_{r}}\frac{\partial h_{r}}%
{\partial\mu_{i}}=0,
\]
so
\[
\frac{\partial h_{r}}{\partial\lambda_{i}}=-\sum_{s=1}^{n}A_{s}^{r}%
\frac{\partial\varphi_{k}}{\partial\lambda_{i}},\ \ \ \frac{\partial h_{r}%
}{\partial\mu_{i}}=-\sum_{s=1}^{n}A_{s}^{r}\frac{\partial\varphi_{k}}%
{\partial\mu_{i}},
\]
where $(A_{s}^{r})$ is a matrix being the inverse of the matrix $(\partial
\varphi_{s}/\partial a_{r})$. In consequence
\begin{align*}
\{h_{r},h_{s}\}  &  =\sum_{k=1}^{n}\left(  \frac{\partial h_{r}}%
{\partial\lambda_{k}}\frac{\partial h_{s}}{\partial\mu_{k}}-\frac{\partial
h_{r}}{\partial\mu_{k}}\frac{\partial h_{s}}{\partial\lambda_{k}}\right)
c_{k}\\
&  =\sum_{k=1}^{n}\left(  \sum_{i,j=1}^{n}A_{i}^{r}\frac{\partial\varphi_{i}%
}{\partial\lambda_{k}}A_{j}^{s}\frac{\partial\varphi_{j}}{\partial\mu_{k}%
}-\sum_{i,j=1}^{n}A_{i}^{r}\frac{\partial\varphi_{i}}{\partial\mu_{k}}%
A_{j}^{s}\frac{\partial\varphi_{j}}{\partial\lambda_{k}}\right)  c_{k}\\
&  =\sum_{i,j=1}^{n}A_{i}^{r}A_{j}^{s}\sum_{k=1}^{n}\left(  \frac
{\partial\varphi_{i}}{\partial\lambda_{k}}\frac{\partial\varphi_{j}}%
{\partial\mu_{k}}-\frac{\partial\varphi_{i}}{\partial\mu_{k}}\frac
{\partial\varphi_{j}}{\partial\lambda_{k}}\right)  c_{k}\\
&  =\sum_{i,j=1}^{n}A_{i}^{r}A_{j}^{s}\{\varphi_{i},\varphi_{j}\}_{\pi}=0,
\end{align*}
where $c_{k}=c(\lambda_{k},\mu_{k})$.
\end{proof}

In result, the system of $n$ evolution equations on $M$
\begin{equation}
\xi_{t_{i}}=\pi dh_{i}=X_{i},\ \ \ i=1,...,n, \label{2.6}%
\end{equation}
where $\xi=(\bm{\lambda},\bm{\mu})^{T},$ is Liouville integrable.

\subsection{Separability\label{s2.2}}

Liouville integrable systems generated by algebraic curves (\ref{2.1}) and the
Poisson tensor (\ref{1.2}) with $c=1$ are separable in the sense of
Hamilton-Jacobi theory, $\left(  \bm{\lambda},\bm{\mu}\right)  $ are then
their separation coordinates and equations (\ref{2.2}) are called separation
relations. Indeed, the Hamiltonian system (\ref{2.6}) is in this case
linearized through a canonical transformation
\begin{equation}
(\bm{\lambda},\bm{\mu})\longrightarrow(\bm{\beta},\bm{\alpha}), \label{3.5}%
\end{equation}
generated by a generating function $W(\bm{\lambda },\bm{\alpha}),$ such that
it satisfies all the Hamilton-Jacobi equations $h_{i}=\alpha_{i}$ of the
system. Then, the transformation (\ref{3.5}) is given implicitly by
\begin{equation}
\beta_{i}=\frac{\partial W}{\partial\alpha_{i}},\ \ \ \mu_{i}=\frac{\partial
W}{\partial\lambda_{i}},\ \ \ \ \ i=1,...,n. \label{3.6}%
\end{equation}
In the variables $(\bm{\beta},\bm{\alpha})$ the $n$ evolution equations
(\ref{2.6}) linearize
\begin{equation}
(\beta_{j})_{t_{i}}=\frac{\partial h_{i}}{\partial\alpha_{j}}=\delta
_{ij},\ \ \ \ \ (\alpha_{j})_{t_{i}}=-\frac{\partial h_{i}}{\partial\beta_{j}%
}=0,\ \ \ \ i,j=1,...,n, \label{3.7}%
\end{equation}
so that
\begin{equation}
\beta_{j}(\bm{\lambda},\bm{\alpha})=\frac{\partial W}{\partial\lambda_{j}%
}=t_{j}+c_{j},\ \ \ \ \ \ c_{j}\in\mathbb{R}. \label{3.8}%
\end{equation}
The existence of separation relations (\ref{2.2}) means that there always
exists an additively separable solution
\begin{equation}
W(\bm{\lambda},\bm{\alpha})=\sum_{i=1}^{n}W_{i}(\lambda_{i},\bm{\alpha}),
\label{3.9}%
\end{equation}
for the generating function $W(\bm{\lambda},\bm{\alpha})$, where functions
$W_{i}(\lambda_{i},\bm{\alpha})$ are solutions of the system of $n$ decoupled
ordinary differential equations
\begin{equation}
\varphi_{i}\left(  \lambda_{i},\frac{dW_{i}(\lambda_{i},\bm{\alpha})}%
{d\lambda_{i}},\bm{\alpha}\right)  =0,\ \ \ \ \ i=1,...,n. \label{3.10}%
\end{equation}

In literature, Liouville integrable systems, linearizable according to
Hamilton-Jacobi theory by additively separable generating function (\ref{3.9})
are known as (generalized) St\"{a}ckel systems. A particularly interesting
case of such systems (a proper St\"{a}ckel system) is generated by $\varphi$
being of hyperelliptic type
\begin{equation}
\varphi(\lambda,\mu)\equiv\sigma(\lambda)+\sum_{k=1}^{n}h_{k}\lambda
^{\gamma_{k}}-\frac{1}{2}f(\lambda)\mu^{2}=0,\ \label{2.7}%
\end{equation}
where $\sigma(\lambda)$ and $f(\lambda)$ are Laurent polynomials in $\lambda$,
$\gamma_{k}\in\mathbb{N}$ and are such that $\gamma_{1}>...>\gamma_{n}=0$.
Then,
\begin{equation}
h_{k}=\frac{1}{2}\bm{\mu}A_{k}G_{f}\bm{\mu}^{T}+V_{k}^{(\sigma)},\qquad
k=1,\ldots,n \label{S3}%
\end{equation}
and some additional geometric structure can be related with the dynamical
systems (\ref{2.6}). The Hamiltonians $h_{k}$ are considered as functions on
the phase space $M=T^{\ast}Q$, where $\bm{\lambda}$ are local coordinates on a
$n$-dimensional configuration space $Q$ and $\bm{\mu}$ are the (fibre)
momentum coordinates, $G_{f}$ is treated as a contravariant metric on $Q$,
defined by the first Hamiltonian $h_{1}$, $A_{k}$ $(A_{1}=I)$ are $(1,1)$-
Killing tensors for the metric $G_{f}$ (for any $f$) and $V_{k}^{(\sigma)}$
are respective potentials on $Q$. The quadratic in momenta $\bm{\mu}$
Hamiltonians \eqref{S3} are in involution with respect to the Poisson operator
$\pi=\sum_{i=1}^{n}\frac{\partial}{\partial\lambda_{i}}\wedge\frac{\partial
}{\partial\mu_{i}}$, in accordance with the general Theorem \ref{KOM}\ in the
subsection above. By construction, the variables $(\bm{\lambda},\bm{\mu})$ are
separation variables for all the St\"{a}ckel Hamiltonians $h_{k}$ in (\ref{S3}).

If we further assume that $\gamma_{k}=n-k$ then the Hamiltonians $h_{k}$ in
(\ref{S3}) become the so-called St\"{a}ckel Hamiltonians of Benenti type
\cite{ben1,ben2,blaszak2007} and in this case
\[
G_{f}=\,\mathrm{diag}\left(  \frac{f(\lambda_{1})}{\Delta_{1}},\ldots
,\frac{f(\lambda_{n})}{\Delta_{n}}\right)  ,\qquad\Delta_{i}=\prod_{j\neq
i}(\lambda_{i}-\lambda_{j}).
\]
Further, $A_{k}$ are given by
\[
A_{k}=(-1)^{k+1}\mathrm{diag}\left(  \frac{\partial s_{k}}{\partial\lambda
_{1}},\ldots,\frac{\partial s_{k}}{\partial\lambda_{n}}\right)  \text{
\ \ \ \ \ }k=1,\ldots,n.
\]
They all are $(1,1)$-Killing tensors for all the metrics $G_{f}$. The function
$s_{k}$ is the elementary symmetric polynomial in $\lambda_{i}$ of degree $k$.
The potential functions $V_{k}^{(\sigma)}$ are given by
\begin{equation}
V_{k}=(-1)^{k+1}\sum_{i=1}^{n}\frac{\partial s_{k}}{\partial\lambda_{i}}%
\frac{\sigma(\lambda_{i})}{\Delta_{i}}\text{ \ \ \ \ }k=1,2,\ldots\label{6.11}%
\end{equation}
and for $\sigma(\lambda)=\sum_{i}a_{i}\lambda^{i}$ take the form $V_{k}%
=\sum_{i}a_{i}\mathcal{V}_{k}^{(i)}$, where the so-called elementary separable
potentials $\mathcal{V}_{k}^{(i)}$ can be explicitly constructed from the
recursion formula \cite{blaszak2011}
\begin{equation}
\mathcal{V}^{(i)}=R^{i}\mathcal{V}^{(0)},\qquad\mathcal{V}^{(i)}%
=\bigl(\mathcal{V}_{1}^{(i)},\ldots,\mathcal{V}_{n}^{(i)}\bigr)^{T}%
,\qquad\mathcal{V}^{(0)}=(0,\ldots,0,-1)^{T}, \label{prec}%
\end{equation}
where
\begin{equation}
R=%
\begin{pmatrix}
-q_{1} & 1 & 0 & 0\\
\vdots & 0 & \ddots & 0\\
\vdots & 0 & 0 & 1\\
-q_{n} & 0 & 0 & 0
\end{pmatrix}
,\qquad R^{-1}=%
\begin{pmatrix}
0 & 0 & 0 & -\frac{1}{q_{n}}\\
1 & 0 & 0 & \vdots\\
0 & \ddots & 0 & \vdots\\
0 & 0 & 1 & -\frac{q_{n-1}}{q_{n}}%
\end{pmatrix}
,\qquad q_{i}\equiv(-1)^{i}s_{i}. \label{V5}%
\end{equation}

\section{Equivalence classes of algebraic curves \label{s3}}

From results of Section \ref{s1} it follows that without loss of generality we
can restrict our construction to Poisson tensors for which coordinates
$(\bm{\lambda},\bm{\mu})$ are canonical coordinates, i.e. when
$c(\bm{\lambda},\bm{\mu})=I$, where $I$ is $n$-dimensional identity matrix. It
follows from the fact that for fixed $c(\bm{\lambda },\bm{\mu})$ we have the
whole family of transformations
\begin{equation}
\bar{\lambda}_{i}=a(\lambda_{i},\mu_{i}),\ \ \ \ \ \bar{\mu}_{i}=b(\lambda
_{i},\mu_{i}),\ \ \ \ i=1,...,n, \label{3.1}%
\end{equation}
where $(\bar{\bm{\lambda}},\bar{\bm{\mu}})$ are canonical coordinates for
Poisson tensor (\ref{2.3}), i.e. they fulfil the condition (\ref{1.6}) for
each pair $(\lambda_{i},\mu_{i})$ of coordinates. Thus, each Liouville
integrable Hamiltonian system (\ref{2.6}), generated by an algebraic curve
(\ref{2.1}) and by the Poisson tensor (\ref{2.3}) with a given $c(\lambda
,\mu)$, can in fact be generated by a whole family (equivalence class) of
algebraic curves $\varphi(\bar{\lambda},\bar{\mu},\bar{h}_{1},...,\bar{h}%
_{n})=0$ and the corresponding Poisson tensors with $c(\bar{\lambda},\bar{\mu
})=1$. Each class represents thus the same dynamical system written in
different Darboux coordinates, related by appropriate canonical transformations.

Let us specify these considerations to the monomial case, following the case
considered in Section \ref{s1}. Consider thus the $2n$-dimensional Hamiltonian
system (\ref{2.6}) generated by the algebraic curve (\ref{2.1}) and by the
Poisson tensor (\ref{2.3}) with $c(\lambda_{i},\mu)=\lambda_{i}^{\alpha}%
\mu_{i}^{\beta}$, with fixed real $\alpha$ and $\beta$. Then, the
transformation to its canonical (Darboux) coordinates on $M$ and its inverse
are of the form (cf. (\ref{1.11}) and (\ref{1.11a}))%
\begin{equation}
\ \bar{\lambda}_{i}=\lambda_{i}^{\bar{\alpha}_{1}}\mu_{i}^{\bar{\alpha}_{2}%
},\ \ \ \ \bar{\mu}_{i}=\lambda_{i}^{\bar{\beta}_{1}}\mu_{i}^{\bar{\beta}_{2}%
},\ \ \ \ \ \lambda_{i}=\bar{\lambda}_{i}^{\bar{\beta}_{2}}\bar{\mu}%
_{i}^{-\bar{\alpha}_{2}},\ \ \ \ \mu_{i}=\bar{\lambda}_{i}^{-\bar{\beta}_{1}%
}\bar{\mu}_{i}^{\bar{\alpha}_{1}}\ \ \ \ i=1,...,n, \label{3.2}%
\end{equation}
where $\bar{\alpha}_{1},\bar{\alpha}_{2},\bar{\beta}_{1},\bar{\beta}_{2}$ are
given either by (\ref{1.8}) (for $\alpha\neq1,$ parameterized by $\bar{\alpha
}_{1}$) or by (\ref{1.9}) (for $\beta\neq1,$ parametrized by $\bar{\beta}_{2}%
$). As a result, our system can be equivalently obtained by either one of the
equivalent curves
\begin{equation}
\varphi(\bar{\lambda},\bar{\mu},\bar{h}_{1},...,\bar{h}_{n})=0 \label{3.3}%
\end{equation}
expressed in coordinates $(\bar{\lambda},\bar{\mu})$ parametrized by
$\bar{\alpha}_{1}$ (in case $\alpha\neq1$) or by $\bar{\beta}_{2}$ (in case
$\beta\neq1$). A canonical transformation and its inverse between coordinates
$(\bar{\bm{\lambda}},\bar{\bm{\mu}})$ and $(\tilde{\bm{\lambda}}%
,\tilde{\bm{\mu}})$ associated with two curves from the class (\ref{3.3})
parametrized by $\bar{\alpha}_{1}$ and by $\tilde{\alpha}_{1}$ respectively,
is given by (cf. (\ref{1.12}) and (\ref{1.12a}))
\begin{equation}
\bar{\lambda}_{i}=\tilde{\lambda}_{i}^{\alpha_{1}}\ \tilde{\mu}_{i}%
^{\alpha_{1}-1},\ \ \ \ \bar{\mu}_{i}=\tilde{\lambda}_{i}^{1-\alpha_{1}}%
\tilde{\mu}^{2-\alpha_{1}},\ \ \ \ \tilde{\lambda}_{i}=\bar{\lambda}%
_{i}^{2-\alpha_{1}}\bar{\mu}_{i}^{1-\alpha_{1}},\ \ \ \ \bar{\mu}_{i}%
=\bar{\lambda}_{i}^{\alpha_{1}-1}\bar{\mu}_{i}^{\alpha_{1}},\ \ \ i=1,...,n,
\label{3.4}%
\end{equation}
where
\[
\ \alpha_{1}=1+\frac{\bar{\alpha}_{1}-\tilde{\alpha}_{1}}{1-\alpha}.
\]

\begin{example}
Let us consider dynamical system (\ref{2.6}), generated by the algebraic curve
(\ref{2.7}) but in the more general case when the Poisson tensor (\ref{1.2})
is given by the monomial $c(\lambda,\mu)=\lambda^{m},\ m\in\mathbb{Z}$ (so
that $\alpha=m$ and $\beta=0$) on the $(\lambda,\mu)$-plane. Such system has
one-parameter family of canonical representations in new coordinates, induced
by the transformation
\begin{equation}
\bar{\lambda}=\lambda^{a}\mu^{\frac{a-1}{1-m}},\ \ \ \bar{\mu}=\lambda
^{1-m-a}\mu^{\frac{2-m-a}{1-m}},\ \ \ \ \ \ a\in\mathbb{R}, \label{4.1}%
\end{equation}
with the inverse%
\[
\lambda=\bar{\lambda}^{\frac{2-m-a}{1-m}}\ \bar{\mu}^{\frac{1-a}{1-m}%
},\ \ \ \mu=\bar{\lambda}^{a+m-1}\bar{\mu}^{a}%
\]
(where we now denote $\bar{\alpha}_{1}$by $a$). Notice that for the
distinguished choice $a=1$ we have $\bar{\lambda}=\lambda,\ \ \ \bar{\mu
}=\lambda^{-m}\mu$ (so that$\ \lambda=\bar{\lambda},\ \ \ \mu=\bar{\lambda
}^{m}\bar{\mu}$) and the algebraic curve (\ref{2.7}) in the new variables
$(\bar{\lambda},\bar{\mu})$ is still of hyperelliptic type
\begin{equation}
\sigma(\bar{\lambda})+\sum_{k=1}^{n}\bar{h}_{k}\bar{\lambda}^{\gamma_{k}%
}=\frac{1}{2}\bar{f}(\bar{\lambda})\bar{\mu}^{2},\ \label{4.2}%
\end{equation}
where $\bar{f}(\bar{\lambda})=f(\bar{\lambda})\bar{\lambda}^{2m}$, with the
canonical Poisson tensor as $c(\bar{\lambda},\bar{\mu})=1$ and thus generates
a St\"{a}ckel system with all Hamiltonians that are again quadratic in momenta%
\begin{equation}
\bar{h}_{k}=\frac{1}{2}\bar{\bm{\mu}}^{T}A_{k}G_{\bar{f}}\bar{\bm{\mu}}%
+V_{k}^{(\sigma)}(\bar{\bm{\lambda}}),\qquad k=1,\ldots,n. \label{4.3}%
\end{equation}
For the choice $a=0$%
\begin{equation}
\bar{\lambda}=\mu^{\frac{1}{m-1}},\ \ \ \bar{\mu}=\lambda^{1-m}\mu^{\frac
{2-m}{1-m}},\ \ \ \ \ \lambda=\bar{\lambda}^{\frac{2-m}{1-m}}\bar{\mu}%
^{\frac{1}{1-m}},\ \ \ \mu=\bar{\lambda}^{m-1}, \label{4.4}%
\end{equation}
so for the particular case $m=2$%
\begin{equation}
\bar{\lambda}=\mu,\ \ \ \bar{\mu}=\lambda^{-1},\ \ \ \ \ \lambda=\bar{\mu
}^{-1},\ \ \ \mu=\bar{\lambda}, \label{4.5}%
\end{equation}
we again obtained the hyperelliptic type curve, this time with interchanged
roles of position and momenta variables:%
\[
\bar{\sigma}(\bar{\mu})+\sum_{k=1}^{n}\bar{h}_{k}\bar{\mu}^{\bar{\gamma}_{k}%
}=\frac{1}{2}\bar{f}(\bar{\mu})\bar{\lambda}^{2},
\]
where%
\begin{equation}
\bar{\sigma}(\bar{\mu})=\sigma(\bar{\mu}^{-1})\ \bar{\mu}^{\gamma_{1}%
},\ \ \ \ \bar{f}(\bar{\mu})=f(\bar{\mu}^{-1})\ \bar{\mu}^{\gamma_{1}%
},\ \ \ \bar{\gamma}_{k}=\gamma_{1}-\gamma_{k},\ \ \ k=1,...,n \label{4.6}%
\end{equation}
and with the normalization $0=\bar{\gamma}_{1}<...<\bar{\gamma}_{n}$.
\end{example}

\section{St\"{a}ckel transform and reciprocal link\label{s4}}

In this section we will assume that the algebraic curve defining Hamiltonian
system depends on a set of $n+n$, instead of just $n$, parameters. We show
that solving this curve with respect to either the first set of $n$ parameters
or the second set of $n$ parameters leads to two integrable systems that can
be related by a St\"{a}ckel transform. We further show that solutions of these
two systems are related by a reciprocal (multi-time) transformation. We
further specify our results to St\"{a}ckel systems.

Consider thus a $2n$-parameter algebraic curve
\begin{equation}
\varphi(\lambda,\mu,a_{1},...,a_{n},b_{1},...,b_{n})=0 \label{6.1}%
\end{equation}
and the corresponding separation relations
\begin{equation}
\varphi(\lambda_{i},\mu_{i},a_{1},...,a_{n},b_{1},...,b_{n})=0,\qquad
i=1,\ldots,n. \label{6.2}%
\end{equation}
Solving these relations with respect to $a_{k}$ (we assume it is possible at
least in some open domain) we obtain $n$ functions (Hamiltonians)
\begin{equation}
a_{k}=h_{k}(\xi,b_{1},...,b_{n}),\ \ \ \ k=1,...,n, \label{6.3}%
\end{equation}
considered on a $2n$-dimensional manifold $M$ (parametrized by coordinates
$\xi=(\bm{\lambda},\bm{\mu})$) and depending on $n$ parameters $b_{1}%
,...,b_{n}.$ These Hamiltonians define $n$ Hamiltonians systems on $M$ of the
form
\begin{equation}
\xi_{t_{i}}=\pi dh_{i}\equiv X_{i},\ \ \ i=1,...,n, \label{6.3a}%
\end{equation}
where $\pi$ is the canonical Poisson tensor of co-rank zero given by
\[
\pi=\sum_{i=1}^{n}\frac{\partial}{\partial\lambda_{i}}\wedge\frac{\partial
}{\partial\mu_{i}}%
\]
and where $t_{1},\ldots,t_{n}$ are respective evolution parameters. The system
(\ref{6.2}) (or equivalently (\ref{6.3}))\ is assumed to be also solvable (at
least in some open domain) with respect to the parameters $b_{k}$ yielding
\begin{equation}
b_{k}=\bar{h}_{k}(\xi,a_{1},...,a_{n}),\ \ \ \ k=1,...,n, \label{6.4}%
\end{equation}
i.e. new Hamiltonian functions $\bar{h}_{k}$ depending on $n$ parameters
$a_{1},...,a_{n}$. The related Hamiltonian systems take the form
\begin{equation}
\xi_{\bar{t}_{i}}=\pi d\bar{h}_{i}\equiv\bar{X}_{i},\ \ \ i=1,...,n,
\label{6.5}%
\end{equation}
where $\bar{t}_{1},...,\bar{t}_{n}$ are respective evolution parameters.

Note that inserting the Hamiltonians $h_{k}$ into the separation curve
(\ref{6.1}) yields the following identity with respect to all $\xi\in M$ and
all $b_{k}$
\begin{equation}
\varphi(\lambda,\mu,h_{1}(\xi,b_{1},...,b_{n}),...,h_{n}(\xi,b_{1}%
,...,b_{n}),b_{1},...,b_{n})\equiv0. \label{6.5a}%
\end{equation}
Similarly, inserting the Hamiltonians $\bar{h}_{k}$ into the separation curve
(\ref{6.1}) yields the following identity with respect to all $\xi\in M$ and
all $a_{k}$
\begin{equation}
\varphi(\lambda,\mu,a_{1},...,a_{n},\bar{h}_{1}(\xi,a_{1},...,a_{n}%
),...,\bar{h}_{n}(\xi,a_{1},...,a_{n}))\equiv0. \label{6.5b}%
\end{equation}

\begin{definition}
The $n$ St\"{a}ckel Hamiltonians $\bar{h}_{k}$ and the $n$ St\"{a}ckel
Hamiltonians $h_{k}$ are called \emph{St\"{a}ckel conjugate Hamiltonians} and
the procedure of mapping $n$ Hamiltonians $h_{k}$ to $n$ Hamiltonians $\bar
{h}_{k}$ is called \emph{St\"{a}ckel transform}.
\end{definition}

Below we remind a theorem explaining the mutual geometric relations between
the St\"{a}ckel conjugate Hamiltonians.

\begin{theorem}
\cite{SAPM2012} For a given $2n$-tuple $(a,b)=(a_{1},...,a_{n},b_{1}%
,\ldots,b_{n})$ of real constants, on the $n$-dimensional submanifold
\begin{equation}
M_{a,b}=\{\xi\in M:\varphi(\lambda_{k},\mu_{k},a,b)=0,\text{ }k=1,...,n\}
\label{6.6a}%
\end{equation}
the following relations hold
\begin{equation}
dh=Ad\bar{h},\ \ \ \ \ X=A\bar{X},\ \ \ \ \ \ \ A_{ij}=-\frac{\partial h_{i}%
}{\partial b_{j}},\text{ \ }i,j=1,\ldots,n, \label{6.7}%
\end{equation}
where $dh=(dh_{1},...,dh_{n})^{T}$, $d\bar{h}=(d\bar{h}_{1},...,d\bar{h}%
_{n})^{T}$, $X=(X_{1},...,X_{n})^{T}$, $\bar{X}=(\bar{X}_{1},...,\bar{X}%
_{n})^{T}.$
\end{theorem}

Note that $M_{a,b}$ can equivalently be defined as%
\begin{align*}
M_{a,b}&=\{\xi\in M:h_{k}(\xi,b_{1},...,b_{n})=a_{k},\ k=1,...,n\}\\
&=\{\xi\in
M:\bar{h}_{k}(\xi,a_{1},...,a_{n})=b_{k},\ k=1,...,n\}
\end{align*}
and that through each point of $M$ there pass infinitely many manifolds
$M_{a,b}$. In fact, fixing all the parameters $b_{k}$ in $M_{a,b}$ and letting
$a_{k}$ vary we obtain a particular foliation of $M$ and, likewise, fixing all
the parameters $a_{k}$ in $M_{a,b}$ and letting $b_{k}$ vary we obtain another
particular foliation of $M$. Note also that each of the manifolds $M_{a,b}$ is
invariant with respect to all $n$ systems (\ref{6.3a}) and all $n$ systems
(\ref{6.5}) which also means that all the vector fields $X_{i}$ and all the
vector fields $\bar{X}_{i}$ are tangent to each manifold $M_{a,b}$. Note that
no relation exists between the vector fields $X$ and $\bar{X}$ on the whole
manifold $M$. Let us also remark that the transformations (\ref{6.7}) on
$M_{a,b}$ can be inverted, yielding
\begin{equation}
d\bar{h}=A^{-1}dh,\ \ \ \ \ \bar{X}=A^{-1}X,\text{ \ with }\left(
A^{-1}\right)  _{ij}=-\frac{\partial\bar{h}_{i}}{\partial a_{j}}\text{,
\ \ }i,j=1,\ldots,n. \label{6.7a}%
\end{equation}
The second relation in (\ref{6.7}) can be reformulated in the dual language,
that of the reciprocal multi-time transformations. The reciprocal
transformation
\[
\bar{t}_{i}=\bar{t}_{i}(t_{1},...,t_{n},\xi_{0}),\ \ \ \ \ i=1,...,n,
\]
given on $M_{b,a}$ by
\begin{equation}
d\bar{t}=A^{T}dt, \label{6.8}%
\end{equation}
where $dt=(dt_{1},...,dt_{n})$ and $d\bar{t}=(d\bar{t}_{1},...,d\bar{t}_{n}),$
transforms the $n$-parameter solutions $\xi(t_{1},...,t_{n},\xi_{0})$ of the
system (\ref{6.3a}) (where $\xi_{0}$ is the initial condition) to the
$n$-parameter solutions $\bar{\xi}(\bar{t}_{1},...,\bar{t}_{n},\xi_{0})$ of
the system (\ref{6.5}). Transformation (\ref{6.8}) is well defined as its
r.h.s is an exact differential.

\begin{definition}
The transformation (\ref{6.8}) between dynamical systems (\ref{6.3a}) and
(\ref{6.5}) on $M_{a,b}$ is called a $n$-parameter reciprocal transform.
\end{definition}

In the remaining part of this section, we will restrict ourselves to the case
of curves (\ref{6.1}) that are affine in all the parameters $a_{i}$ and
$b_{i}$. In this case the relations (\ref{6.3}) take the form
\begin{equation}
a_{k}=h_{k}(\xi,b_{1},...,b_{n})\equiv H_{k}+\sum_{j=1}^{n}H_{k}^{(j)}%
b_{j},\ \ \ \ \ k=1,...,n, \label{6.8a}%
\end{equation}
while the relations (\ref{6.4})\ attain the form%
\begin{equation}
b_{k}=\bar{h}_{k}(\xi,a_{1},...,a_{n})=\bar{H}_{k}+\sum_{j=1}^{n}\bar{H}%
_{k}^{(j)}a_{j},\ \ \ \ k=1,...,n. \label{6.8c}%
\end{equation}
The St\"{a}ckel transform between the Hamiltonians $h_{k}$ and $\bar{h}_{k}$
takes the explicit matrix form
\begin{equation}
h=A(\bar{H}-b)\text{ or }\bar{h}=A^{-1}(H-a), \label{6.8b}%
\end{equation}
where $h=(h_{1},...,h_{n})^{T}$, $b=(b_{1},...,b_{n})^{T}$, $H=(H_{1}%
,...,H_{n})^{T}$, $\bar{h}=(\bar{h}_{1},\ldots,\bar{h}_{n})^{T}$,
$a=(a_{1},...,a_{n})^{T}$ and $A_{ij}=-\frac{\partial h_{i}}{\partial b_{j}%
}=-H_{i}^{(j)}$. Note that after setting all the $a_{i}$ and $b_{i}$ equal to
zero we obtain the following matrix formula relating Hamiltonians $H_{k}$ and
$\bar{H}_{k}$%
\begin{equation}
H=A\bar{H}, \label{6.8d}%
\end{equation}
where $\bar{H}=(\bar{H}_{1},...,\bar{H}_{n})^{T}$, valid on the whole $M$.
Formula (\ref{6.8d}) is the parameter-independent part of the St\"{a}ckel
transform between $h_{k}$ and $\bar{h}_{k}$.

Consider now a specification of the above affine case when the separation
curve (\ref{6.1}) attains the following hyperelliptic-type form
\begin{equation}
\sigma(\lambda)+\sum_{j=1}^{n}b_{j}\lambda^{\gamma_{j}}+\sum_{k=1}^{n}%
a_{k}\lambda^{n-k}=\frac{1}{2}f(\lambda)\mu^{2}, \label{6.9}%
\end{equation}
where $\sigma(\lambda)$ and $f(\lambda)$ are Laurent polynomials in $\lambda,$
$\gamma_{1}>...>\gamma_{n}$ are natural numbers. Solving the corresponding
separation relations with respect to $a_{k}$ yields the separable systems
belonging to the Benenti subclass of St\"{a}ckel systems. Explicitly, we
obtain $n$ quadratic in momenta St\"{a}ckel Hamiltonians
\begin{equation}
h_{k}=H_{k}+\sum_{j=1}^{n}b_{j}\mathcal{V}_{k}^{(\gamma_{j})}=\frac{1}%
{2}\mathbb{\mu}^{T}A_{k}G_{f}\mathbb{\mu}+V_{k}^{(\sigma)}+\sum_{j=1}^{n}%
b_{j}\mathcal{V}_{k}^{(\gamma_{j})},\qquad k=1,\ldots,n. \label{6.10}%
\end{equation}
The structure and geometric meaning of the Hamiltonians $h_{k}$ is as those
described in subsection \ref{s2.2}.

Performing the St\"{a}ckel transform on the set of $n$ Hamiltonians
(\ref{6.10}) we obtain the set of $n$ Hamiltonians $\bar{h}_{k}$ of the form%
\begin{equation}
\bar{h}_{k}=\bar{H}_{k}+\sum_{j=1}^{n}a_{j}\bar{\mathcal{V}}_{k}^{(n-j)}%
=\frac{1}{2}\mathbb{\mu}^{T}\bar{A}_{k}\bar{G}_{f}\mathbb{\mu}+\bar{V}%
_{k}^{(\sigma)}+\sum_{j=1}^{n}a_{j}\bar{\mathcal{V}}_{k}^{(n-j)},\qquad
k=1,\ldots,n, \label{5a}%
\end{equation}
(where $\bar{G}_{f}$ is defined by $\bar{H}_{1}$ with $\bar{A}_{1}=I$)
generated by the separation curve%
\begin{equation}
\sigma(\lambda)+\sum_{j=1}^{n}a_{j}\lambda^{n-j}+\sum_{k=1}^{n}\bar{h}%
_{k}\lambda^{\gamma_{k}}=\frac{1}{2}f(\lambda)\mu^{2}. \label{4}%
\end{equation}
The Hamiltonians $\bar{h}_{k}$ define the Hamiltonian evolution equations
(\ref{6.5}). Then, on $n$-dimensional submanifold (\ref{6.6a}) the relations
(\ref{6.7}) hold with
\begin{equation}
A_{kj}=-\frac{\partial h_{k}}{\partial b_{j}}=-\mathcal{V}_{k}^{(\gamma_{j})}
\label{8}%
\end{equation}
and the relations (\ref{6.7a}) hold with
\[
(A^{-1})_{kj}=-\frac{\partial\bar{h}_{k}}{\partial a_{j}}=-\bar{\mathcal{V}%
}_{k}^{(n-j)}.
\]

\begin{remark}
From the above considerations it follows that systems generated by algebraic
curves (\ref{4}) can always be transformed, by an appropriate reciprocal
transformation, to systems from Benenti class, generated by algebraic curves
(\ref{6.9}).
\end{remark}

The St\"{a}ckel systems generated by curves of the type (\ref{4}) have been
thoroughly studied in \cite{blaszak2005}.

\begin{example}
\label{E1}Consider the algebraic curve (\ref{6.9}) of the form%
\begin{equation}
b_{1}\lambda^{2}+a_{1}\lambda+b_{2}+a_{2}=\frac{1}{2}\lambda\mu^{2}%
+\lambda^{4}. \label{ziuta}%
\end{equation}
Solving the corresponding separation relations with respect to $a_{k}$ we
obtain Hamiltonians $h_{1}$ and $h_{2}$ as in (\ref{6.8a}). In Vi\'{e}te
coordinates $(q,p)$, associated with separable coordinates \thinspace
$(\lambda,\mu)$ through the point transformation%
\begin{align*}
q_{1}  &  =-\lambda_{1}-\lambda_{2}\text{,
\ \ \ \ \ \ \ \ \ \ \ \ \ \ \ \ \ \ \ \ \ }q_{2}=\lambda_{1}\lambda_{2}\text{,
\ }\\
p_{1}  &  =\frac{1}{\lambda_{2}-\lambda_{1}}(\lambda_{1}\mu_{1}-\lambda_{2}%
\mu_{2})\text{, \ }p_{2}=\frac{1}{\lambda_{2}-\lambda_{1}}(\mu_{1}-\mu_{2}),
\end{align*}
the Hamiltonians $h_{k}$ attain the form%
\begin{align*}
h_{1}  &  =\frac{1}{2}p_{1}^{2}-\frac{1}{2}q_{2}p_{2}^{2}-q_{1}^{3}%
+2q_{1}q_{2}+b_{1}q_{1},\\
h_{2}  &  =-q_{2}p_{1}p_{2}-\frac{1}{2}q_{1}q_{2}p_{2}^{2}-q_{1}^{2}%
q_{2}+q_{2}^{2}+b_{1}q_{2}-b_{2}.
\end{align*}
Passing to flat coordinates \cite{blaszak2007} $(x,y)$ defined through the
point transformation%
\begin{equation}
q_{1}=-x_{1}\text{, }q_{2}=-\frac{1}{4}x_{2}^{2}\text{, \ }p_{1}=-y_{1},\text{
\ }p_{2}=-\frac{2}{x_{2}}y_{2}, \label{flat}%
\end{equation}
we obtain $h_{k}$ in the form%
\begin{align}
h_{1}  &  =\frac{1}{2}y_{1}^{2}+\frac{1}{2}y_{2}^{2}+x_{1}^{3}+\frac{1}%
{2}x_{1}x_{2}^{2}-b_{1}x_{1},\label{hh}\\
h_{2}  &  =\frac{1}{2}x_{2}y_{1}y_{2}-\frac{1}{2}x_{1}y_{2}^{2}+\frac{1}%
{4}x_{1}^{2}x_{2}^{2}+\frac{1}{16}x_{2}^{4}-\frac{1}{4}b_{1}x_{2}^{2}%
-b_{2}.\nonumber
\end{align}
Note that for $b_{1}=b_{2}=0$ the Hamiltonians $h_{1}$ and $h_{2}$ represent
one of the integrable cases of H\'{e}non-Heiles systems \cite{fordy1991}. The
matrix $A$ in (\ref{8})\ and its inverse attain in the $(x,y)$-variables the
form%
\[
A=\left(
\begin{array}
[c]{ll}%
x_{1} & 0\\
\frac{1}{4}x_{2}^{2} & 1
\end{array}
\right)  \text{, \ \ }A^{-1}=\left(
\begin{array}
[c]{cc}%
\frac{1}{x_{1}} & 0\\
-\frac{1}{4x_{1}}x_{2}^{2} & 1
\end{array}
\right)  .
\]
Solving the separation relations corresponding to (\ref{ziuta}) with respect
to $b_{k}$ yields the Hamiltonians $\bar{h}_{1},\bar{h}_{2}$ that in the flat
coordinates (\ref{flat}) attain the form%
\begin{align*}
\bar{h}_{1}  &  =\frac{1}{2x_{1}}y_{1}^{2}+\frac{1}{2x_{1}}y_{2}^{2}+x_{1}%
^{2}+\frac{1}{2}x_{2}^{2}-\frac{a_{1}}{x_{1}},\\
\bar{h}_{2}  &  =-\frac{x_{2}^{2}}{8x_{1}}y_{1}^{2}+\left(  -\frac{1}{2}%
x_{1}-\frac{1}{8}\frac{x_{2}^{2}}{x_{1}}\right)  y_{2}^{2}+\frac{1}{2}%
x_{2}y_{1}y_{2}-\frac{1}{16}x_{2}^{4}+\frac{1}{4}a_{1}\frac{x_{2}^{2}}{x_{1}%
}-a_{2}.
\end{align*}
The Hamiltonians $\bar{h}_{1},\bar{h}_{2}$ and the Hamiltonians $h_{1},h_{2}$
are St\"{a}ckel conjugate. Note that the variables (\ref{flat}) are only
conformally flat for the Hamiltonians $\bar{h}_{1},\bar{h}_{2}$. The manifolds
$M_{a,b}$ are given by
\[
M_{a,b}=\left\{  (x,y):h_{1}(x,y,b_{1},b_{2})=a_{1},\text{ \ }h_{2}%
(x,y,b_{1},b_{2})=a_{2}\right\}
\]
or by%
\[
M_{a,b}=\left\{  (x,y):\bar{h}_{1}(x,y,a_{1},a_{2})=b_{1},\text{ \ }\bar
{h}_{2}(x,y,a_{1},a_{2})=b_{2}\right\}
\]
and one can verify by a direct computation that on $M_{a,b}$ we have
$X=A\bar{X}$ as well as $\bar{X}=A^{-1}X$. The corresponding reciprocal
transformation (\ref{6.8}) between the evolution parameters takes the form%
\[
d\bar{t}_{1}=x_{1}dt_{1}+\frac{1}{4}x_{2}^{2}dt_{2}\text{, \ }d\bar{t}%
_{2}=dt_{2}.
\]

\end{example}

St\"{a}ckel transform was first described by J. Hietarinta et al in
\cite{hietarinta} (where it was called the coupling-constant metamorphosis)
and in \cite{boyeretal}. In this early approach this transform was only
one-parameter. In its most general form St\"{a}ckel transform has been
introduced in \cite{serg2008} and then intensively studied in
\cite{SAPM2012,blasz2019}.

\section{Miura maps\label{s5}}

In this section\ we investigate yet another possibility of generating
integrable and separable Hamiltonian systems from algebraic curves. We will
consider algebraic curves depending on $n+N$ parameters having a certain
block-type structure. These curves generate integrable and separable
Hamiltonian systems that can be connected by a finite-dimensional analogue of
Miura maps, known from soliton theory (Theorem \ref{obvious}, see also its
proof in the Appendix). These finite-dimensional Miura maps yield in turn
multi-Hamiltonian formulation of the obtained integrable systems (Theorem
\ref{obvious2}). Results of this section generalize the results for the
one-block case, obtained earlier in \cite{marciniak2023} as well as the
results obtained in \cite{blaszak2009}.

Consider thus the $(n+N)$-parameter algebraic curve
\begin{equation}
\varphi(\lambda,\mu,a_{1},...,a_{n},c_{1},...,c_{N})=0, \label{5.1}%
\end{equation}
with $1\leq N\leq n,$ in the following form
\begin{equation}
\varphi_{0}(\lambda,\mu)+\sum_{k=1}^{m}\varphi_{k}(\lambda,\mu)\psi
_{k}(\lambda,a_{1}^{(k)},...,a_{n_{k}\,}^{(k)},c_{1}^{(k)},...,c_{\alpha
}^{(k)})=0,\ \ \ n_{1}+...+n_{m}=n,\ \ \ \alpha m=N, \label{5.1a}%
\end{equation}
where $1\leq\alpha\leq\min(n_{k})$ and where
\begin{equation}
\psi_{k}(\lambda,a_{1}^{(k)},...,a_{n_{k}}^{(k)},c_{1}^{(k)},...,c_{\alpha
}^{(k)})=c_{\alpha}^{(k)}\lambda^{n_{k}-1+\alpha}+...+c_{1}^{(k)}%
\lambda^{n_{k}}+\sum_{i=1}^{n_{k}}a_{i}^{(k)}\lambda^{n_{k}-i} \label{5.1b}%
\end{equation}
with the normalization $\varphi_{m}(\lambda,\mu)=1$.\ The curve (\ref{5.1a})
consists thus of $m$ blocks of Benenti type. Solving the related separation
relations with respect to $a_{i}^{(k)}$ we obtain $n$ Hamiltonian functions
\begin{equation}
a_{i}^{(k)}=h_{i}^{(k)}(\xi,c_{1}^{(1)},...,c_{\alpha}^{(m)}%
),\ \ \ \ \ i=1,...,n_{k}\text{, \ \ }k=1,...,m\ \ \ \label{5.3}%
\end{equation}
on a $2n$-dimensional open submanifold $M\mathcal{\subset}\bm{R}^{2n}$
parametrized by$\ \xi=\left(  \bm{\lambda},\bm{\mu}\right)$. Assume now that $\bm{c}=\left(
c_{1}^{(1)},\ldots,c_{\alpha}^{(1)},\ldots,c_{1}^{(m)},\ldots,c_{\alpha}%
^{(m)}\right)  $ are additional coordinates on the $(2n+N)$-dimensional open
submanifold $\mathcal{M\subset}\bm{R}^{2n+\alpha m}$, parametrized by
$(\bm{\lambda},\bm{\mu,c})$. The Hamiltonians $h_{i}^{(k)}$ in (\ref{5.3})
generate $n$ dynamical Hamiltonian systems (a St\"{a}ckel system) on
$\mathcal{M}$, given by
\begin{equation}
\xi_{t_{i}^{(k)}}=\pi_{0}dh_{i}^{(k)}(\bm{\lambda},\bm{\mu,c})\equiv
X_{i}^{(k)},\ \ \ i=1,...,n_{k},\ \ k=1,...,m, \label{5.4}%
\end{equation}
where $\xi\in\mathcal{M}$, $\pi_{0}$ is the canonical Poisson tensor
\[
\pi_{0}=\sum_{i=1}^{n}\frac{\partial}{\partial\lambda_{i}}\wedge\frac
{\partial}{\partial\mu_{i}}%
\]
of co-rank $\alpha m=N$ on $\mathcal{M}$ and $c_{1}^{(1)},...,c_{\alpha}%
^{(m)}$ are its Casimir functions.

The goal of this section is to construct a Miura map between the St\"{a}ckel
system (\ref{5.4}), generated by the curve (\ref{5.1a}), (\ref{5.1b}), and the
St\"{a}ckel system%
\begin{equation}
\xi_{t_{i}^{(k)}}=\bar{\pi}d\bar{h}_{r}^{(k)}%
(\bm{\bar{\lambda}},\bm{\bar{\mu},\bar{c}})\equiv\bar{X}_{i}^{(k)}%
\text{,\ \ }i=1,...,n_{k},\ \ k=1,...,m, \label{5.4a}%
\end{equation}
where%
\begin{equation}
\bar{\pi}=%
{\textstyle\sum_{i=1}^{n}}
\frac{\partial}{\partial\bar{\lambda}_{i}}\wedge\frac{\partial}{\partial
\bar{\mu}_{i}}, \label{pibar}%
\end{equation}
generated by the curve
\begin{equation}
\varphi_{0}(\bar{\lambda},\bar{\lambda}^{s}\bar{\mu})\bar{\lambda}^{-s}%
+\sum_{k=1}^{m}\varphi_{k}(\bar{\lambda},\bar{\mu}\bar{\lambda}^{s})\psi
_{k}(\bar{\lambda},\bar{h}_{1}^{(k)},...,\bar{h}_{n_{k}\,}^{(k)},\bar{c}%
_{1}^{(k)},...,\bar{c}_{\alpha}^{(k)})=0, \label{5.7}%
\end{equation}
with $n_{1}+...+n_{m}=n, \alpha \cdot m=N$, where
\begin{align}
&\psi_{k}(\bar{\lambda},\bar{h}_{1}^{(k)},...,\bar{h}_{n_{k}}^{(k)},\bar{c}%
_{1}^{(k)},...,\bar{c}_{\alpha}^{(k)})\\ \nonumber
&=\bar{c}_{\alpha}^{(k)}\bar{\lambda
}^{n_{k}+\alpha-s-1}+\ldots+\bar{c}_{s+1}^{(k)}\bar{\lambda}^{n_{k}}%
+\sum_{i=1}^{n_{k}}\bar{h}_{i}^{(k)}\bar{\lambda}^{n_{k}-i}+\bar{c}_{s}%
^{(k)}\bar{\lambda}^{-1}+\ldots+\bar{c}_{1}^{(k)}\bar{\lambda}^{-s},
\label{5.8}%
\end{align}
$s$ is an integer such that $1\leq s\leq\alpha$ and where the coordinates
$(\bm{\bar{\lambda}},\bm{\bar{\mu}},\bm{\bar{c}})=(\bar{\lambda}_{i},\bar{\mu
}_{i},\bar{c}_{1}^{(1)},\ldots,\bar{c}_{\alpha}^{(1)},\\ \ldots,$ $\bar{c}%
_{1}^{(m)},\ldots,\bar{c}_{\alpha}^{(m)})_{i=1,...,n}$ on $\mathcal{M}$ are
some functions of coordinates $(\bm{\lambda},\bm{\mu,c})$.

Consider the following map in $\bm{R}^{2}$:%
\begin{equation}
\bar{\lambda}=\lambda,\text{ \ \ }\bar{\mu}=\lambda^{-s}\mu. \label{5.9}%
\end{equation}
This map transforms (algebraically) the curve (\ref{5.1a}) into the curve
(\ref{5.7}), provided that for all $k=1,...,m$%
\begin{equation}%
\begin{array}
[c]{ll}%
\bar{c}_{i}^{(k)}=h_{n_{k}-i+1}^{(k)}, & i=1,\ldots,s\\
\bar{c}_{i}^{(k)}=c_{i}^{(k)}, & i=s+1,\ldots,\alpha\\
\bar{h}_{i}^{(k)}=c_{s-i+1}^{(k)}, & i=1,\ldots,s\\
\bar{h}_{i}^{(k)}=h_{i-s}^{(k)}, & i=s+1,\ldots,n_{k},
\end{array}
\label{5.10}%
\end{equation}
The relations (\ref{5.10}) can be inverted to%
\begin{equation}%
\begin{array}
[c]{ll}%
c_{i}^{(k)}=\bar{c}_{i}^{(k)}, & i=s+1,\ldots,\alpha\\
c_{i}^{(k)}=\bar{h}_{s-i+1}^{(k)}, & i=1,\ldots,s\\
h_{i}^{(k)}=\bar{h}_{s+i}^{(k)}, & i=1,\ldots,n_{k}-s\\
h_{i}^{(k)}=\bar{c}_{n-i+1}^{(k)}, & i=n_{k}-s+1,\ldots,n_{k}.
\end{array}
\label{mapinv}%
\end{equation}
The maps (\ref{5.9}) and (\ref{5.10}) induce the following Miura maps
$\mathfrak{M}:\mathcal{M}\rightarrow\mathcal{M}$%

\begin{equation}%
\begin{array}
[c]{ll}%
\bar{\lambda}_{i}=\lambda_{i}, & i=1,\ldots,n\\
\bar{\mu}_{i}=\lambda_{i}^{-s}\mu_{i}, & i=1,\ldots,n\\
\bar{c}_{i}^{(k)}=h_{n_{k}-i+1}^{(k)}(\lambda,\mu,c), & i=1,\ldots,s\\
\bar{c}_{i}^{(k)}=c_{i}^{(k)}, & i=s+1,\ldots,\alpha,
\end{array}
\label{miura}%
\end{equation}
with the inverse $\mathfrak{M}^{-1}:\mathcal{M}\rightarrow\mathcal{M}$%
\begin{equation}%
\begin{array}
[c]{ll}%
\lambda_{i}=\bar{\lambda}_{i}, & i=1,\ldots,n\\
\mu_{i}=\bar{\lambda}_{i}^{s}\bar{\mu}_{i}, & i=1,\ldots,n\\
c_{i}^{(k)}=\bar{h}_{s-i+1}^{(k)}(\bar{\lambda},\bar{\mu},\bar{c}), &
i=1,\ldots,s\\
c_{i}^{(k)}=\bar{c}_{i}^{(k)}, & i=s+1,\ldots,\alpha.
\end{array}
\label{minv}%
\end{equation}
Let us now present the main theorem of this section.

\begin{theorem}
\label{obvious}For any $s\in\{1,\ldots,\alpha\}$ the $n$ Hamiltonian vector
fields $X_{i}^{(k)}$ in (\ref{5.4}) and the $n$ Hamiltonian vector fields
$\bar{X}_{i}^{(k)}$ in (\ref{5.4a}) pairwise coincide, provided that the
coordinates $(\bm{\bar{\lambda}},\bm{\bar{\mu}},\bm{\bar{c}})$\ and
$(\bm{\lambda},\bm{\mu,c})$ are connected by the Miura map (\ref{miura}):
\[
X_{i}^{(k)}=\bar{X}_{i}^{(k)},\text{ \ \ }i=1,\ldots,n_{k},\text{
\ \ }k=1,\ldots,m.
\]

\end{theorem}

The proof of this theorem can be found in Appendix. This theorem means that
all the St\"{a}ckel systems (\ref{5.4a}), generated by the curves (\ref{5.7})
(one for each value of $s$ between $1$ and $\alpha$), represent on the
extended phase space $\mathcal{M}$ the same St\"{a}ckel system as the
St\"{a}ckel system (\ref{5.4}), (the one generated by the curve (\ref{5.1a})),
written in different coordinates, connected by the corresponding invertible
Miura maps (\ref{miura}). We can thus call the St\"{a}ckel system (\ref{5.4a})
an $s$-representation of the St\"{a}ckel system (\ref{5.4}), with
$s=0,\ldots,\alpha$ (where $s=0$-representation means simply the original
St\"{a}ckel system (\ref{5.4})). Since all the Miura maps (\ref{miura}) are
invertible it also means that there exists direct Miura maps (appropriate
compositions of (\ref{miura}) and (\ref{minv})) between different
$s$-representations of our St\"{a}ckel system, see \cite{marciniak2023}.

An important consequence of the above construction is the following theorem,
that generalizes the corresponding one-block theorem from \cite{marciniak2023}.

\begin{theorem}
\label{obvious2}The St\"{a}ckel system (\ref{5.4}) is $(\alpha+1)$%
-Hamiltonian, i.e. for all $s=0,\ldots,\alpha$ (and for all $k=1,\ldots,m$)
\begin{align}
X_{i}^{(k)}  &  =\pi_{0}dh_{i}^{(k)}=\pi_{s}d^{(k)}c_{s-i+1},\text{
\ \ }i=1,\ldots s,\label{bHX}\\
X_{i}^{(k)}  &  =\pi_{0}dh_{i}^{(k)}=\pi_{s}dh_{i-s}^{(k)},\text{
\ }i=s+1,\ldots n_{k},\nonumber
\end{align}
where
\begin{equation}
\pi_{s}=\sum_{i=1}^{n}\lambda_{i}^{s}\frac{\partial}{\partial\lambda_{i}%
}\wedge\frac{\partial}{\partial\mu_{i}}\text{ }+\sum_{k=1}^{m}\sum_{j=1}%
^{s}X_{j}^{(k)}\wedge\frac{\partial}{\partial c_{s-j+1}^{(k)}}%
,\ \ \ \ \ s=0,...,\alpha. \label{5.11}%
\end{equation}

\end{theorem}

Thus, the matrix representation of $\pi_{s}$ in the variables
$(\bm{\lambda},\bm{\mu},\bm{c})$ is given by the following $(n+\alpha
m)\times(n+\alpha m)$ matrix:
\begin{equation}
\pi_{s}(\bm{\lambda},\bm{\mu},\bm{c})=\left(
\begin{array}
[c]{cc}%
\begin{array}
[c]{ll}%
\ \ 0 & \Lambda^{s}\\
-\Lambda^{s} & 0
\end{array}
& \text{\ }X_{s}^{(1)}\ldots\ X_{1}^{(1)}\overset{\alpha-s}{\text{
\ }\overbrace{0...0}\ }\cdots\overset{}{\ X_{s}^{(m)}\ldots\ X_{1}%
^{(m)}\overset{\alpha-s}{\text{ \ }\overbrace{0...0}}}\\
\ast & \ \ \ \ \ \ \ \ \ \ \ \ \ \ \ \ \ \
\end{array}
\right)  , \label{5.11m}%
\end{equation}
where $\Lambda=\mathrm{diag}(\lambda_{1},\ldots,\lambda_{n})$, $X_{i}^{(k)}$
denote here the columns consisting of \emph{components} of the vector field
$X_{i}^{(k)}$ in the coordinates $(\bm{\lambda},\bm{\mu},\bm{c})$ and where
$\ast$ denotes transpositions of the corresponding $X_{i}^{(k)}$. The proof of
this theorem is obtained by a direct computation of the Poisson operator
$\bar{\pi}=%
{\textstyle\sum_{i=1}^{n}}
\partial_{\bar{\lambda}_{i}}\wedge\partial_{\bar{\mu}_{i}}$, for each and
every case $s=1,\ldots,\alpha$, in the variables $(\bm{\lambda},\bm{\mu,c})$
associated with the $s=0$ representation (for the proof of the one-block
version of this theorem, see \cite{marciniak2023}).

Using the notation
\[
h_{j}^{(k)}\equiv c_{1-j}^{(k)}\text{ \ for }j=0,-1,\ldots,-\alpha+1,
\]
we can write the formulas (\ref{bHX}) in the more compact form%
\begin{equation}
X_{i}^{(k)}=\pi_{s}dh_{i-s}^{(k)},\text{ \ \ }s=0,\ldots,\alpha. \label{bH}%
\end{equation}
Using this notation, we can formulate the following corollary.

\begin{corollary}
\bigskip The multi-Hamiltonian representations (\ref{bH}) generate, for each \newline
$k\in\left\{  1,\ldots,m\right\}  $, $\binom{\alpha+1}{2}$ bi-Hamiltonian
chains
\begin{equation}%
\begin{array}
[c]{l}%
\pi_{i}dh_{-i}^{(k)}\ \ \ =0\\
\pi_{i}dh_{-i+1}^{(k)}=X_{1}^{(k)}\ =\pi_{j}dh_{-j+1}^{(k)}\\
\ \ \ \ \ \ \ \ \ \ \ \ \ \ \ \ \vdots\\
\pi_{i}dh_{-i+r}^{(k)}=X_{r}^{(k)}\ =\pi_{j}dh_{-j+r}^{(k)}\\
\ \ \ \ \ \ \ \ \ \ \ \ \ \ \ \ \vdots\\
\pi_{i}dh_{-i+n_{k}}^{(k)}=X_{n_{k}}^{(k)}=\pi_{j}dh_{-j+n_{k}}^{(k)}\\
\ \ \ \ \ \ \ \ \ \ \ \ \ \ \ \ \ \ \ \ 0=\pi_{j}dh_{-j+n_{k}+1}^{(k)}%
\end{array}
\label{bh}%
\end{equation}
for $0\leq i<j\leq\alpha$.
\end{corollary}

There are two limit cases of the above construction. The first one is the case
when $m=n$ (so that all blocks have length one: $n_{k}=1$ for all
$k=1,\ldots,m)$. Then, the vector fields $X_{i}^{(k)}$ in (\ref{5.4}) on
$\mathcal{M}$ are only bi-Hamiltonian, forming $n$ one-field chains%
\[%
\begin{array}
[c]{l}%
\pi_{0}dc_{1}^{(k)}\ =\ 0\\
\pi_{0}dh_{1}^{(k)}=X_{1}^{(k)}=\pi_{1}dc_{1}^{(k)}\\
\ \ \ \ \ \ \ \ \ \ \ \ \ \ \ \ 0\ \ =\pi_{1}dh_{1}^{(k)}%
\end{array}
,\ \ \ \ \ k=1,...,n,
\]
where
\[
\pi_{0}=\sum_{i=1}^{n}\frac{\partial}{\partial\lambda_{i}}\wedge\frac
{\partial}{\partial\mu_{i}}\text{, \ \ }\pi_{1}=\sum_{i=1}^{n}\lambda_{i}%
\frac{\partial}{\partial\lambda_{i}}\wedge\frac{\partial}{\partial\mu_{i}%
}\text{ }+\sum_{k=1}^{n}X_{1}^{(k)}\wedge\frac{\partial}{\partial c_{1}^{(k)}%
}.
\]
This particular situation was considered in \cite{blaszak2009}. The opposite
limit case takes place when $m=1$ (i.e. when there is only one block in the
curve (\ref{5.1a})). Then, the considered St\"{a}ckel system is $(n+1)$%
-Hamiltonian%
\[
X_{i}=\pi_{s}dh_{i-s},\text{ \ \ }s=0,\ldots,n,
\]
with
\[
\pi_{0}=\sum_{i=1}^{n}\frac{\partial}{\partial\lambda_{i}}\wedge\frac
{\partial}{\partial\mu_{i}}\text{, \ \ }\pi_{s}=\sum_{i=1}^{n}\lambda_{i}%
^{s}\frac{\partial}{\partial\lambda_{i}}\wedge\frac{\partial}{\partial\mu_{i}%
}\text{ }+\sum_{j=1}^{s}X_{j}\wedge\frac{\partial}{\partial c_{s-j+1}%
},\ \ \ \ s=0,...,n
\]
(and with the notation $X_{i}^{(1)}\equiv X_{i},\ i=1,...,n$). In this case
there is in total $\binom{N+1}{2}$ bi-Hamiltonian chains of the form
(\ref{bh}), where $1\leq N\leq n$. This situation was considered in
\cite{marciniak2023}.

\begin{example}
\label{E2}Consider the special case of curve (\ref{ziuta}) from Example
\ref{E1} with $b_{1}=b_{2}=0$ but in the space extended by the coordinates
$c_{1}$ and $c_{2}$:%
\begin{equation}
c_{2}\lambda^{3}+c_{1}\lambda^{2}+h_{1}\lambda_{1}+h_{2}=\frac{1}{2}\lambda
\mu^{2}+\lambda^{4}. \label{jola}%
\end{equation}
This curve has a form of (\ref{5.1a})-(\ref{5.1b}) with $n=2,m=1$ (so it is a
one-block case)$,\alpha=2$ (so that $N=2$), $\varphi_{0}=-\frac{1}{2}%
\lambda\mu^{2}-\lambda^{4}$, $\varphi_{1}=1$. The Miura map (\ref{miura}) for
$s=1$ attains the form
\[
\bar{\lambda}_{1}=\lambda_{1},\text{ \ }\bar{\lambda}_{2}=\lambda_{2},\text{
\ }\bar{\mu}_{1}=\lambda_{1}^{-1}\mu_{1},\text{ \ }\bar{\mu}_{2}=\lambda
_{2}^{-1}\mu_{2}\text{, \ }\bar{c}_{1}=h_{2}(\lambda,\mu,c),\text{ }\bar
{c}_{2}=c_{2}%
\]
and it transforms the St\"{a}ckel system generated by the curve (\ref{jola})
to the St\"{a}ckel system generated by the curve
\begin{equation}
\bar{c}_{2}\bar{\lambda}^{2}+\bar{h}_{1}\bar{\lambda}+\bar{h}_{2}+\bar{c}%
_{1}\bar{\lambda}^{-1}=\frac{1}{2}\bar{\lambda}^{2}\bar{\mu}^{2}+\bar{\lambda
}^{3}. \label{jola1}%
\end{equation}
Further, for $s=2$ the Miura map (\ref{miura}) attains the form
\[
\widetilde{\lambda}_{1}=\lambda_{1},\text{ }\widetilde{\lambda}_{2}%
=\lambda_{2},\text{ \ }\widetilde{\mu}_{1}=\lambda_{1}^{-2}\mu_{1},\text{
}\widetilde{\mu}_{2}=\lambda_{2}^{-2}\mu_{2}\text{, \ }\widetilde{c}_{1}%
=h_{2}(\lambda,\mu,c),\text{ }\widetilde{c}_{2}=h_{1}(\lambda,\mu,c)
\]
(in this example we have to distinguish between the two set of "bar"
variables, one for $s=1$ and for $s=2$ so in the latter case we use the
variables $\left(  \widetilde{\lambda},\widetilde{\mu},\widetilde{c}\right)
$) and it transforms the St\"{a}ckel system generated by the curve
(\ref{jola}) to the St\"{a}ckel system generated by the curve
\begin{equation}
\widetilde{h}_{1}\widetilde{\lambda}+\widetilde{h}_{2}+\text{ }\widetilde
{c}_{2}\widetilde{\lambda}^{-1}+\widetilde{c}_{1}\widetilde{\lambda}%
^{-2}=\frac{1}{2}\widetilde{\lambda}^{3}\widetilde{\mu}^{2}+\widetilde
{\lambda}. \label{jola2}%
\end{equation}
All these three St\"{a}ckel systems are three different $s$-representations
(with $s=0,1$ and $2$) of the same three-Hamiltonian St\"{a}ckel system, with
its three Poisson tensors (\ref{5.11m}) having in the $(\mathbf{\lambda
},\mathbf{\mu},\mathbf{c})$-variables the form
\[
\pi_{0}(\bm{\lambda},\bm{\mu},\bm{c})=\left(
\begin{array}
[c]{cc}%
\begin{array}
[c]{ll}%
\ \ 0 & I\\
-I & 0
\end{array}
& 0\text{ } \ 0\\
\ast &
\end{array}
\right)  ,
\]%
\[
\pi_{1}(\bm{\lambda},\bm{\mu},\bm{c})=\left(
\begin{array}
[c]{cc}%
\begin{array}
[c]{ll}%
\ \ 0 & \Lambda\\
-\Lambda & 0
\end{array}
& X_{1}\text{ }0\\
\ast &
\end{array}
\right)  ,
\]%
\[
\pi_{2}(\bm{\lambda},\bm{\mu},\bm{c})=\left(
\begin{array}
[c]{cc}%
\begin{array}
[c]{ll}%
\ \ 0 & \Lambda^{2}\\
-\Lambda^{2} & 0
\end{array}
& X_{2}\text{ }X_{1}\\
\ast &
\end{array}
\right)  ,
\]
where $I=\mathrm{diag}(1,1)$, $\Lambda=\mathrm{diag}(\lambda_{1},\lambda_{2}%
)$. Let us write these objects explicitly in the flat coordinates
(\ref{flat}). The two St\"{a}ckel Hamiltonians $h_{1}$ and $h_{2}$ have in
these coordinates the form%
\begin{align*}
h_{1}  &  =\frac{1}{2}y_{1}^{2}+\frac{1}{2}y_{2}^{2}+x_{1}^{3}+\frac{1}%
{2}x_{1}x_{2}^{2}-c_{2}\left(  x_{1}^{2}+\frac{1}{4}x_{2}^{2}\right)
-c_{1}x_{1},\\
h_{2}  &  =\frac{1}{2}x_{2}y_{1}y_{2}-\frac{1}{2}x_{1}y_{2}^{2}+\frac{1}%
{4}x_{1}^{2}x_{2}^{2}+\frac{1}{16}x_{2}^{4}-\frac{1}{4}c_{2}x_{1}x_{2}%
^{2}-\frac{1}{4}c_{1}x_{2}^{2}%
\end{align*}
while the matrix representations of the Poisson operators $\pi_{0},\pi_{1}$
and $\pi_{2}$ attain the explicit form%
\[
\pi_{0}(\bm{x},\bm{y},\bm{c})=\left(
\begin{array}
[c]{cc}%
\begin{array}
[c]{ll}%
\ \ 0 & I\\
-I & 0
\end{array}
& 0\text{ } \ 0\\
\ast &
\end{array}
\right)  ,
\]%
\[
\pi_{1}(\bm{x},\bm{y},\bm{c})=\left(
\begin{array}
[c]{cc}%
\begin{array}
[c]{cccc}%
0 & 0 & x_{1} & \frac{1}{2}x_{2}\\
0 & 0 & \frac{1}{2}x_{2} & 0\\
-x_{1} & -\frac{1}{2}x_{2} & 0 & \frac{1}{2}y_{2}\\
-\frac{1}{2}x_{2} & 0 & -\frac{1}{2}y_{2} & 0
\end{array}
& X_{1}\text{ \ }0\\
\ast & \text{ \ \ \ \ \ }
\end{array}
\right)  ,
\]%
\[
\pi_{2}(\bm{x},\bm{y},\bm{c})=\left(
\begin{array}
[c]{c}%
\begin{array}
[c]{cccc}%
0 & 0 & x_{1}^{2}+\frac{1}{4}x_{2}^{2} & \frac{1}{2}x_{1}x_{2}\\
0 & 0 & \frac{1}{2}x_{1}x_{2} & \frac{1}{4}x_{2}^{2}\\
-x_{1}^{2}-\frac{1}{4}x_{2}^{2} & -\frac{1}{2}x_{1}x_{2} & 0 & \frac{1}%
{2}x_{1}y_{2}\\
-\frac{1}{2}x_{1}x_{2} & -\frac{1}{4}x_{2}^{2} & -\frac{1}{2}x_{1}y_{2} & 0
\end{array}%
\begin{array}
[c]{cc}
& X_{2}\text{ }X_{1}\\
&
\end{array}
\\
\ast
\end{array}
\right)  ,
\]
where the components of the vector fields $X_{1}$ and $X_{2}$ are given by
\begin{align*}
X_{1}    = & \left(  y_{1},y_{2},-3x_{1}^{2}-\frac{1}{2}x_{2}^{2}+2c_{2}%
x_{1}+c_{1},-x_{1}x_{2}+\frac{1}{2}c_{2}x_{2},0,0\right)  ^{T},\\
& \\
X_{2}   = & \left(  \frac{1}{2}x_{2}y_{2},\frac{1}{2}x_{2}y_{1}-x_{1}%
y_{2},\frac{1}{2}y_{2}^{2}-\frac{1}{2}x_{1}x_{2}^{2}+\frac{1}{4}c_{2}x_{2}%
^{2},\right. \\
& \left. -\frac{1}{2}y_{1}y_{2}-\frac{1}{2}x_{1}^{2}x_{2}-\frac{1}{8}x_{2}%
^{3}+\frac{1}{2}c_{2}x_{1}x_{2}+\frac{1}{2}c_{1}x_{2},0,0\right)
^{T},\text{\ }%
\end{align*}
while the corresponding bi-Hamiltonian chains are%
\[%
\begin{array}
[c]{l}%
\pi_{0}dc_{1}=0\\
\pi_{0}dh_{1}=X_{1}=\pi_{1}dc_{1}\\
\pi_{0}dh_{2}=X_{2}=\pi_{1}dh_{1}\\
\ \ \ \ \ \ \ \ \ \ \ \ \ 0=\pi_{1}dh_{2}%
\end{array}
\ \ \ \ \
\begin{array}
[c]{l}%
\pi_{0}dc_{1}=0\\
\pi_{0}dh_{1}=X_{1}=\pi_{2}dc_{2}\\
\pi_{0}dh_{2}=X_{2}=\pi_{2}dc_{1}\\
\ \ \ \ \ \ \ \ \ \ \ \ \ 0=\pi_{2}dh_{1}%
\end{array}
\ \ \ \ \
\begin{array}
[c]{l}%
\pi_{1}dc_{2}=0\\
\pi_{1}dc_{1}=X_{1}=\pi_{2}dc_{2}\\
\pi_{1}dh_{1}=X_{2}=\pi_{2}dc_{1}\\
\ \ \ \ \ \ \ \ \ \ \ \ \ 0=\pi_{2}dh_{1}%
\end{array}
\ .
\]

\end{example}

\begin{example}
\label{E3}Consider now another specification of the curve (\ref{ziuta}) from
Example \ref{E1}, this time with $a_{1}=a_{2}=0$ and in the space extended by
the coordinates $c^{(1)}$ and $c^{(2)}$. More specifically, we consider the
curve%
\[
\lambda^{2}\left(  c_{1}^{(2)}\lambda+h_{1}^{(2)}\right)  +c_{1}^{(1)}%
\lambda+h_{1}^{(1)}=\frac{1}{2}\lambda\mu^{2}+\lambda^{4}.
\]
This curve has the form of (\ref{5.1a})-(\ref{5.1b}) with $n=2,m=2$. Thus, it
is a two-block case, with $n_{1}=n_{2}=1$, $\alpha=1$ (so that $N=2$),
$\varphi_{0}=-\frac{1}{2}\lambda\mu^{2}-\lambda^{4}$, $\varphi_{1}=1$ and
$\varphi_{2}=\lambda^{2}$. \ To simplify the notation, let us denote
$h_{1}=h_{1}^{(2)}$,$h_{2}=$ $h_{1}^{(1)}$, $c_{1}=c_{1}^{(2)}$, $c_{2}%
=c_{1}^{(1)}$ so that the curve is%
\begin{equation}
\lambda^{2}\left(  c_{1}\lambda+h_{1}\right)  +c_{2}\lambda+h_{2}=\frac{1}%
{2}\lambda\mu^{2}+\lambda^{4}. \label{ela1}%
\end{equation}
The Miura map (\ref{miura}) for $s=1$ attains the form
\[
\bar{\lambda}_{1}=\lambda_{1},\text{ \ }\bar{\lambda}_{2}=\lambda_{2},\text{
\ }\bar{\mu}_{1}=\lambda_{1}^{-1}\mu_{1},\text{ \ }\bar{\mu}_{2}=\lambda
_{2}^{-1}\mu_{2}\text{, \ }\bar{c}_{1}=h_{1}(\lambda,\mu,c)\text{, \ }\bar
{c}_{2}=h_{2}(\lambda,\mu,c)
\]
and it transforms the St\"{a}ckel system generated by the curve (\ref{ela1})
to the St\"{a}ckel system generated by the curve
\begin{equation}
\bar{\lambda}^{2}\left(  \bar{h}_{1}+\bar{c}_{1}\bar{\lambda}^{-1}\right)
+\bar{h}_{2}+\bar{c}_{2}\bar{\lambda}^{-1}=\frac{1}{2}\bar{\lambda}^{2}%
\bar{\mu}^{2}+\bar{\lambda}^{3}. \label{ela2}%
\end{equation}
Both St\"{a}ckel systems are two different $s$-representations (with $s=0$%
,$1$) of the same bi-Hamiltonian St\"{a}ckel system. Two Poisson tensors
(\ref{5.11m}) have in the $(\bm{\lambda},\bm{\mu},\bm{c})$-variables the form
\[
\pi_{0}(\bm{x},\bm{y},\bm{c})=\left(
\begin{array}
[c]{cc}%
\begin{array}
[c]{ll}%
\ \ 0 & I\\
-I & 0
\end{array}
& 0\text{ } \ 0\\
\ast &
\end{array}
\right)  ,
\]%
\[
\pi_{1}(\bm{\lambda},\bm{\mu},\bm{c})=\left(
\begin{array}
[c]{cc}%
\begin{array}
[c]{ll}%
\ \ 0 & \Lambda\\
-\Lambda & 0
\end{array}
& X_{1}\text{ }X_{2}\\
\ast &
\end{array}
\right)  ,
\]
where $X_{i}$ denote here the components of the vector fields $\pi_{0}dh_{i}$.
Let us write down these objects explicitly in the conformally flat coordinates
(\ref{flat}). The St\"{a}ckel Hamiltonians $h_{1}$ and $h_{2}$ are%
\begin{align*}
h_{1}  &  =\frac{1}{2x_{1}}y_{1}^{2}+\frac{1}{2x_{1}}y_{2}^{2}+x_{1}^{2}%
+\frac{1}{2}x_{2}^{2}-c_{1}\left(  x_{1}+\frac{x_{2}^{2}}{4x_{1}}\right)
-\frac{c_{2}}{x_{1}},\\
h_{2}  &  =-\frac{x_{2}^{2}}{8x_{1}}y_{1}^{2}+\left(  -\frac{1}{2}x_{1}%
-\frac{1}{8}\frac{x_{2}^{2}}{x_{1}}\right)  y_{2}^{2}+\frac{1}{2}x_{2}%
y_{1}y_{2}-\frac{1}{16}x_{2}^{4}+\frac{1}{16}c_{1}\frac{x_{2}^{4}}{x_{1}%
}+\frac{1}{4}c_{2}\frac{x_{2}^{2}}{x_{1}},
\end{align*}
while the matrix representations of the Poisson operators $\pi_{0}$ and
$\pi_{1}$ attain the explicit form%
\[
\pi_{0}(\bm{x},\bm{y},\bm{c})=\left(
\begin{array}
[c]{cc}%
\begin{array}
[c]{ll}%
\ \ 0 & I\\
-I & 0
\end{array}
& 0\text{ } \ 0\\
\ast &
\end{array}
\right)  ,
\]%
\[
\pi_{1}(\bm{x},\bm{y},\bm{c})=\left(
\begin{array}
[c]{cc}%
\begin{array}
[c]{cccc}%
0 & 0 & x_{1} & \frac{1}{2}x_{2}\\
0 & 0 & \frac{1}{2}x_{2} & 0\\
-x_{1} & -\frac{1}{2}x_{2} & 0 & \frac{1}{2}y_{2}\\
-\frac{1}{2}x_{2} & 0 & -\frac{1}{2}y_{2} & 0
\end{array}
& X_{1}\text{ \ }X_{2}\\
\ast & \text{ \ \ \ \ \ }%
\end{array}
\right)  ,
\]
where%
\begin{align*}
X_{1}    =& \left(  \frac{y_{1}}{x_{1}},\frac{y_{2}}{x_{1}},\frac{1}{2x_{1}%
^{2}}y_{1}^{2}+\frac{1}{2x_{1}^{2}}y_{2}^{2}-2x_{1}+c_{1}\left(  1-\frac
{x_{2}^{2}}{4x_{1}^{2}}\right)  ,-x_{2}+\frac{1}{2}c_{1}\frac{x_{2}}{x_{1}%
},0,0\right)  ^{T},\\
X_{2}    =& \left(  -\frac{1}{4}\frac{x_{2}^{2}}{x_{1}}y_{1}+\frac{1}{2}%
x_{2}y_{2},-\left(  x_{1}+\frac{1}{4}\frac{x_{2}^{2}}{x_{1}^{2}}\right)
y_{2}+\frac{1}{2}x_{2}y_{1},\right. \\
& -\frac{1}{8}\frac{x_{2}^{2}}{x_{1}^{2}}y_{1}%
^{2}-\left(  -\frac{1}{2}+\frac{1}{8}\frac{x_{2}^{2}}{x_{1}^{2}}\right)
y_{2}^{2}+\frac{1}{16}c_{1}\frac{x_{2}^{4}}{x_{1}^{2}}+\frac{1}{4}c_{2}%
\frac{x_{2}^{2}}{x_{1}^{2}}, \\
&  \left.  \frac{1}{4}\frac{x_{2}}{x_{1}}(y_{1}^{2}+y_{2}%
^{2})-\frac{1}{2}y_{1}y_{2}+\frac{1}{4}x_{2}^{3}-\frac{1}{4}c_{1}\frac
{x_{2}^{3}}{x_{1}}-\frac{1}{2}c_{2}\frac{x_{2}}{x_{1}},0,0\right)  ^{T}.
\end{align*}
The two corresponding bi-Hamiltonian chains are%
\[%
\begin{array}
[c]{l}%
\pi_{0}dc_{2}\ =\ 0\\
\pi_{0}dh_{2}\ =X_{2}=\pi_{1}dc_{2}\\
\ \ \ \ \ \ \ \ \ \ \ \ 0\ \ =\pi_{1}dh_{2}%
\end{array}
,\ \ \ \
\begin{array}
[c]{l}%
\pi_{0}dc_{1}\ =\ 0\\
\pi_{0}dh_{1}\ =X_{1}=\pi_{1}dc_{1}\\
\ \ \ \ \ \ \ \ \ \ \ \ 0\ \ =\pi_{1}dh_{1}.
\end{array}
\
\]

\end{example}

\section*{Appendix}

\setcounter{equation}{0} \renewcommand{\theequation}{A.\arabic{equation}} We
prove here Theorem \ref{obvious}. Let us fix $k\in\left\{  1,\ldots,m\right\}
$ and then $r\in\left\{  1,\ldots,n_{k}\right\}  $. We want to show that the
vector fields $X_{r}^{(k)}$ and $\bar{X}_{r}^{(k)}$ on $\mathcal{M}$ coincide.
Obviously, $\bar{X}_{r}^{(k)}=\bar{\pi}d\bar{h}_{r}^{(k)}=%
{\textstyle\sum\limits_{i=1}^{n}}
\left(  \frac{\partial\bar{h}_{r}^{(k)}}{\partial\bar{\mu}_{i}}\frac{\partial
}{\partial\bar{\lambda}_{i}}-\frac{\partial\bar{h}_{r}^{(k)}}{\partial
\bar{\lambda}_{i}}\frac{\partial}{\partial\bar{\mu}_{i}}\right)  $. Let us
write this vector field in the coordinates $(\lambda_{i},\mu_{i},c_{1}%
^{(1)},\ldots,c_{\alpha}^{(1)},\ldots,c_{1}^{(m)},\ldots,c_{\alpha}%
^{(m)})_{i=1,...,n}$ connected with the coordinates \newline
$(\bar{\lambda}_{i} \bar{\mu}_{i},\bar{c}_{1}^{(1)}\ldots, \bar{c}_{\alpha}^{(1)}, \ldots,\bar
{c}_{1}^{(m)},\ldots,\bar{c}_{\alpha}^{(m)})_{i=1,...,n}$ through the Miura
map (\ref{minv}). Components of $\bar{X}_{r}^{(k)}$ in the non-bar coordinates
will be given by
\begin{equation}
J\left(
\begin{array}
[c]{c}%
\frac{\partial\bar{h}_{r}^{(k)}}{\partial\bar{\mu}}\\
-\frac{\partial\bar{h}_{r}^{(k)}}{\partial\bar{\lambda}}\\
0_{N\times1}%
\end{array}
\right)  , \label{Ap1}%
\end{equation}
where
\[
\frac{\partial\bar{h}_{r}^{(k)}}{\partial\bar{\mu}}=\left(  \frac{\partial
\bar{h}_{r}^{(k)}}{\partial\bar{\mu}_{1}},\ldots,\frac{\partial\bar{h}%
_{r}^{(k)}}{\partial\bar{\mu}_{n}}\right)  ^{T}\text{, \ \ \ }\frac
{\partial\bar{h}_{r}^{(k)}}{\partial\bar{\lambda}}=\left(  \frac{\partial
\bar{h}_{r}^{(k)}}{\partial\bar{\lambda}_{1}},\ldots,\frac{\partial\bar{h}%
_{r}^{(k)}}{\partial\bar{\lambda}_{n}}\right)  ^{T}%
\]
and where $J$ is the Jacobian of the map (\ref{minv}), given explicitly by%

\[
J=\left(
\begin{array}
[c]{cc}%
\begin{array}
[c]{ccccc}%
I_{n} &  &  &  & 0_{n\times n}\\
&  &  &  & \\
s\Lambda^{s-1}U &  &  &  & \Lambda^{s}%
\end{array}
& 0_{2n\times N}\\%
\begin{array}
[c]{c}%
B_{1}\\
\vdots\\
B_{m}%
\end{array}
& \ast
\end{array}
\right)  ,
\]
with $\Lambda=$\textrm{diag}$(\bar{\lambda}_{1},\ldots,\bar{\lambda}_{n})$,
$U=$\textrm{diag}$(\bar{\mu}_{1},\ldots,\bar{\mu}_{n})$, and with $B_{l}$,
$l=1,\ldots,m$ being the $\alpha\times2n$ matrix given by%
\[
B_{l}=\left(
\begin{array}
[c]{cc}%
\left[  \frac{\partial\bar{h}_{s-i+1}^{(l)}}{\partial\bar{\lambda}_{j}%
}\right]  _{\substack{i=1\ldots s\\j=1,\ldots n}} & \left[  \frac{\partial
\bar{h}_{s-i+1}^{(l)}}{\partial\bar{\mu}_{j}}\right]  _{\substack{i=1\ldots
s\\j=1,\ldots n}}\\
& \\
0_{(\alpha-s)\times n} & 0_{(\alpha-s)\times n}%
\end{array}
\right)
\]
while $\ast$ represent some $N\times N$ matrix. Thus, (\ref{Ap1}) has the form%
\[
J\left(
\begin{array}
[c]{c}%
\frac{\partial\bar{h}_{r}^{(k)}}{\partial\bar{\mu}}\\
-\frac{\partial\bar{h}_{r}^{(k)}}{\partial\bar{\lambda}}\\
0_{N\times1}%
\end{array}
\right)  =\left(
\begin{array}
[c]{c}%
\frac{\partial\bar{h}_{r}^{(k)}}{\partial\bar{\mu}}\\
s\Lambda^{s-1}U\frac{\partial\bar{h}_{r}^{(k)}}{\partial\bar{\mu}}-\Lambda
^{s}\frac{\partial\bar{h}_{r}^{(k)}}{\partial\bar{\lambda}}\\
Z_{1}\\
\vdots\\
Z_{m}%
\end{array}
\right)
\]
where%
\[
Z_{l}=\left(
\begin{array}
[c]{c}%
\left\{  \bar{h}_{s}^{(l)},\bar{h}_{r}^{(k)}\right\}  _{\bar{\pi}}\\
\vdots\\
\left\{  \bar{h}_{1}^{(l)},\bar{h}_{r}^{(k)}\right\}  _{\bar{\pi}}\\
0_{(\alpha-s)\times1}%
\end{array}
\right)  =0_{\alpha\times1},
\]
since all the Hamiltonians $\bar{h}_{r}^{(k)}$ mutually Poisson commute with
respect to $\bar{\pi}$. So, (\ref{Ap1}) is given by%
\[
J\left(
\begin{array}
[c]{c}%
\frac{\partial\bar{h}_{r}^{(k)}}{\partial\bar{\mu}}\\
-\frac{\partial\bar{h}_{r}^{(k)}}{\partial\bar{\lambda}}\\
0_{N\times1}%
\end{array}
\right)  =\left(
\begin{array}
[c]{c}%
\frac{\partial\bar{h}_{r}^{(k)}}{\partial\bar{\mu}}\\
s\Lambda^{s-1}U\frac{\partial\bar{h}_{r}^{(k)}}{\partial\bar{\mu}}-\Lambda
^{s}\frac{\partial\bar{h}_{r}^{(k)}}{\partial\bar{\lambda}}\\
0_{N\times1}%
\end{array}
\right)  .
\]
Therefore, $X_{r}^{(k)}=\bar{X}_{r}^{(k)}$ provided that
\begin{equation}
\frac{\partial\bar{h}_{r}^{(k)}}{\partial\bar{\mu}_{p}}=\frac{\partial
h_{r}^{(k)}}{\partial\mu_{p}}\text{, \ }p=1,\ldots,n \label{Ap2}%
\end{equation}
and provided that%
\[
s\bar{\lambda}_{p}^{s-1}\bar{\mu}_{p}\frac{\partial\bar{h}_{r}^{(k)}}%
{\partial\bar{\mu}_{p}}-\bar{\lambda}_{p}^{s}\frac{\partial\bar{h}_{r}^{(k)}%
}{\partial\bar{\lambda}_{p}}=-\frac{\partial h_{r}^{(k)}}{\partial\lambda_{p}%
},\text{ \ \ }p=1,\ldots,n.
\]
Using (\ref{Ap2}) and the fact that $\bar{\lambda}_{p}=\lambda_{p}$, $\bar
{\mu}_{p}=\lambda_{i}^{-p}\mu_{i},$, this last condition is equivalent to%
\begin{equation}
\frac{\partial\bar{h}_{r}^{(k)}}{\partial\bar{\lambda}_{p}}=\lambda_{p}%
^{-s}\frac{\partial h_{r}^{(k)}}{\partial\lambda_{p}}+s\lambda_{p}^{-s-1}%
\mu_{p}\frac{\partial h_{r}^{(k)}}{\partial\mu_{p}}\text{, \ }p=1,\ldots,n.
\label{Ap3}%
\end{equation}
We will now prove first (\ref{Ap2}) and then (\ref{Ap3}). The separation
relations following from the curve (\ref{5.1a}) yield the following $n$
identities on $\mathcal{M\subset}\bm{R}^{2n+\alpha m}$:
\begin{equation}
\varphi_{0}(\lambda_{i},\mu_{i})+\sum_{k=1}^{m}\varphi_{k}(\lambda_{i},\mu
_{i})\psi_{k}(\lambda_{i},h_{1}^{(k)},...,h_{n_{k}\,}^{(k)},c_{1}%
^{(k)},...,c_{\alpha}^{(k)})\equiv0,\ \ \ i=1,\ldots,n, \label{Id}%
\end{equation}
where now
\begin{equation}
\psi_{k}(\lambda_{i},h_{1}^{(k)},...,h_{n_{k}}^{(k)},c_{1}^{(k)}%
,...,c_{\alpha}^{(k)})=c_{\alpha}^{(k)}\lambda_{i}^{n_{k}-1+\alpha}%
+...+c_{1}^{(k)}\lambda_{i}^{n_{k}}+\sum_{j=1}^{n_{k}}h_{j}^{(k)}\lambda
_{i}^{n_{k}-j}. \label{psi}%
\end{equation}
Differentiating each of these identities with respect to $\mu_{p}$ yields (no
summation over $i$) we obtain%
\begin{equation}
\varphi_{0,2}^{\prime}(\lambda_{i},\mu_{i})\delta_{ip}+\sum_{k=1}^{m}%
\varphi_{k,2}^{\prime}(\lambda_{i},\mu_{i})\psi_{k}(\ldots)\delta_{ip}%
+\sum_{k=1}^{m}\varphi_{k}(\lambda_{i},\mu_{i})\sum_{j=1}^{n_{k}}%
\frac{\partial h_{j}^{(k)}}{\partial\mu_{p}}\lambda_{i}^{n_{k}-j}%
\equiv0\text{, \ }\ i=1,\ldots,n. \label{jeden}%
\end{equation}
Let us note, for later purposes, that for $i\neq p$ the above identities
attain the form%
\begin{equation}
\sum_{k=1}^{m}\varphi_{k}(\lambda_{i},\mu_{i})\sum_{j=1}^{n_{k}}\frac{\partial
h_{j}^{(k)}}{\partial\mu_{p}}\lambda_{i}^{n_{k}-j}\equiv0\text{, \ }\ i\neq p.
\label{jedenprim}%
\end{equation}
Analogously, the separation relations following from the curve (\ref{5.7})
yield the following $n$ identities on $\mathcal{M}$:%
\begin{equation}
\varphi_{0}(\bar{\lambda}_{i},\bar{\lambda}_{i}^{s}\bar{\mu}_{i})\bar{\lambda
}_{i}^{-s}+\sum_{k=1}^{m}\varphi_{k}(\bar{\lambda}_{i},\bar{\mu}_{i}%
\bar{\lambda}_{i}^{s})\psi_{k}(\bar{\lambda}_{i},\bar{h}_{1}^{(k)},...,\bar
{h}_{n_{k}\,}^{(k)},\bar{c}_{1}^{(k)},...,\bar{c}_{\alpha}^{(k)}%
)\equiv0,\ \ \ i=1,\ldots,n, \label{Idbar}%
\end{equation}
where now
\begin{align}  \label{psibar}
\psi_{k}(\bar{\lambda}_{i},\bar{h}_{1}^{(k)},...,\bar{h}_{n_{k}}^{(k)},\bar
{c}_{1}^{(k)},...,\bar{c}_{\alpha}^{(k)})= & \bar{c}_{\alpha}^{(k)}\bar{\lambda
}_{i}^{n_{k}+\alpha-s-1}+\ldots+\bar{c}_{s+1}^{(k)}\bar{\lambda}_{i}^{n_{k}%
} \nonumber \\ 
& + \sum_{j=1}^{n_{k}}\bar{h}_{j}^{(k)}\bar{\lambda}_{i}^{n_{k}-j}+\bar{c}%
_{s}^{(k)}\bar{\lambda}_{i}^{-1}+\ldots+\bar{c}_{1}^{(k)}\bar{\lambda}%
_{i}^{-s}. 
\end{align}
Differentiating each of these identities with respect to $\bar{\mu}_{p}$
yields (again, no summation over $i$)%
\begin{equation}
\varphi_{0,2}^{\prime}(\bar{\lambda}_{i},\bar{\lambda}_{i}^{s}\bar{\mu}%
_{i})\delta_{ip}+\sum_{k=1}^{m}\varphi_{k,2}^{\prime}(\bar{\lambda}_{i}%
,\bar{\lambda}_{i}^{s}\bar{\mu}_{i})\bar{\lambda}_{i}^{s}\psi_{k}(\bar{\left.
\ldots\right.  })\delta_{ip}+\sum_{k=1}^{m}\varphi_{k}(\bar{\lambda}_{i}%
,\bar{\lambda}_{i}^{s}\bar{\mu})\sum_{j=1}^{n_{k}}\frac{\partial\bar{h}%
_{j}^{(k)}}{\partial\bar{\mu}_{p}}\lambda_{i}^{n_{k}-j}\equiv0,\text{
\ } \label{dwa}%
\end{equation}
$i=1,\ldots,n$, where $\psi_{k}(\bar{\left.  \ldots\right.  })$ is the short-hand notation for
(\ref{psibar}). A simple identification through (\ref{mapinv}) shows that
$\psi_{k}(...)$ in (\ref{psi}) and $\psi_{k}(\bar{\left.  \ldots\right.  })$
in (\ref{psibar}) coincide as functions on $\mathcal{M}$. Thus, the Miura map
(\ref{minv}) maps the identities (\ref{dwa}) exactly onto the corresponding
(i.e. with the same $i$) identities (\ref{jeden}). That means that
$\frac{\partial h_{j}^{(k)}}{\partial\mu_{p}}$ and $\frac{\partial\bar{h}%
_{j}^{(k)}}{\partial\bar{\mu}_{p}}$ satisfy the same set of $n$ linear
equations, with the same non-degenerated system matrix, and thus they must
pairwise coincide. Thus, (\ref{Ap2}) is proven.

Let us now prove (\ref{Ap3}). Differentiating all the identities (\ref{Id})
with respect to $\lambda_{p}$ we obtain, for $i=p$%
\begin{align}
&  \varphi_{0,1}^{\prime}(\lambda_{p},\mu_{p})+\sum_{k=1}^{m}\varphi
_{k,1}^{\prime}(\lambda_{p},\mu_{p})\psi_{k}(\lambda_{p},\ldots)\nonumber\\
&  +\sum_{k=1}^{m}\varphi_{k}(\lambda_{p},\mu_{p})\left(  \sum_{j=1}^{\alpha
}(n_{k}+j-1)c_{j}^{(k)}\lambda_{p}^{n_{k}+j-2}+\sum_{j=1}^{n_{k}}%
(n_{k}-j)\lambda_{p}^{n_{k}-j-1}h_{j}^{(k)} \right. \nonumber \\
& \hspace{80pt} \left. +\sum_{j=1}^{n_{k}}\lambda
_{p}^{n_{k}-j}\frac{\partial h_{j}^{(k)}}{\partial\lambda_{p}}\right) =0
\label{ip}
\end{align}
and for $i\neq p$%

\begin{equation}
\sum_{k=1}^{m}\varphi_{k}(\lambda_{i},\mu_{i})\sum_{j=1}^{n_{k}}\lambda
_{i}^{n_{k}-j}\frac{\partial h_{j}^{(k)}}{\partial\lambda_{p}}\equiv0.
\label{inp}%
\end{equation}
Analogously, multiplying all the identities (\ref{Idbar}) by $\bar{\lambda
}_{i}^{s}$\ and differentiating with respect to $\bar{\lambda}_{p}$ we obtain,
for $i=p$%
\begin{align}
&  \varphi_{0,1}^{\prime}(\bar{\lambda}_{p},\bar{\mu}_{p}\bar{\lambda}_{p}%
^{s})+\varphi_{0,2}^{\prime}(\bar{\lambda}_{p},\bar{\mu}_{p}\bar{\lambda}%
_{p}^{s})s\bar{\mu}_{p}\bar{\lambda}_{p}^{s-1}+\sum_{k=1}^{m}\varphi
_{k,1}^{\prime}(\bar{\lambda}_{p},\bar{\mu}_{p}\bar{\lambda}_{p}^{s})\psi
_{k}(\bar{\lambda}_{p},\bar{\left.  \ldots\right.  }) \label{ip bar} \\
& +\sum_{k=1}^{m} \varphi_{k,2}^{\prime}(\bar{\lambda}_{p},\bar{\mu}_{p}\bar{\lambda}_{p}%
^{s})s\bar{\mu}_{p}\bar{\lambda}_{p}^{s-1}\psi_{k}(\bar{\lambda}_{p}%
,\bar{\left.  \ldots\right.  }) \nonumber\\
&+\sum_{k=1}^{m}\varphi_{k}(\bar{\lambda}_{p},\bar{\mu}_{p}\bar{\lambda}%
_{p}^{s})\left(  \sum_{j=1}^{\alpha}(n_{k}+j-1)\bar{c}_{j}^{(k)}\bar{\lambda
}_{p}^{n_{k}+j-2}  +\sum_{j=1}^{s}(j-1)\bar{c}_{j}^{(k)}\bar{\lambda}_{p}%
^{j-2}\right. \nonumber \\
&\hspace{100pt} + \left. \sum_{j=1}^{n_{k}}(n_{k}+s-j)\bar{\lambda}_{p}^{n_{k}+s-j-1}\bar{h}%
_{j}^{(k)} +  \sum_{j=1}^{n_{k}}\bar{\lambda}_{p}^{n_{k}+s-j}\frac{\partial
\bar{h}_{j}^{(k)}}{\partial\bar{\lambda}_{p}}\right)  \equiv0 \nonumber
\end{align}
and for $i\neq p$%
\begin{equation}
\sum_{k=1}^{m}\varphi_{k}(\bar{\lambda}_{i},\bar{\mu}_{i}\bar{\lambda}_{i}%
^{s})\sum_{j=1}^{n_{k}}\bar{\lambda}_{i}^{n_{k+s-j}}\frac{\partial\bar{h}%
_{j}^{(k)}}{\partial\bar{\lambda}_{p}}=0. \label{inpbar}%
\end{equation}
It is immediate to see that the Miura map (\ref{minv}) transforms all the
relations (\ref{inpbar}) to the corresponding (i.e. with the same $i$)
relations (\ref{inp}). Further, careful comparison of all terms in (\ref{ip})
and (\ref{ip bar}) using (\ref{jedenprim}) shows that the Miura map
(\ref{minv}) transforms (\ref{ip bar})\ onto (\ref{ip}) if and only if the
condition (\ref{Ap3}) holds. Thus, (\ref{Ap3}) is proved.

\label{lastpage}
\end{document}